\documentclass[prx,twocolumn,showpacs, superscriptaddress,preprintnumbers,amssymb]{revtex4-2}

\usepackage{bm}
\usepackage[dvipdfmx]{graphicx}
\usepackage{amsmath}
\usepackage{amssymb}
\usepackage{amsthm}
\usepackage{braket}
\usepackage{siunitx}
\usepackage{color}
\usepackage[utf8]{inputenc}
\usepackage{hyperref}
\usepackage{here}
\usepackage{ytableau} 
\usepackage{xcolor}

\usepackage{tikz}
\usepackage{tikz-cd}
\usetikzlibrary{arrows}
\usetikzlibrary{intersections}
\usetikzlibrary{shapes.geometric}
\usetikzlibrary{decorations.pathmorphing, patterns,shapes}
\usetikzlibrary{decorations.markings}

\usepackage[most]{tcolorbox}

\hypersetup{ 
setpagesize=false,
 bookmarksnumbered=true,%
 bookmarksopen=true,%
 colorlinks=true,%
 linkcolor=blue,
 citecolor=blue,
 urlcolor=blue
}

\theoremstyle{plain}
\newtheorem{thm}{Theorem}
\newtheorem{lm}{Lemma}

\newtheorem*{thm*}{Theorem}
\theoremstyle{definition}

\usepackage{physics}
\usepackage{mathtools}

\begin{document}
\title{Scaling-optimal purification of noisy qubit unitary channels}

\author{Ryotaro Niwa}
\affiliation{Department of Physics, Graduate School of Science, The University of Tokyo, Hongo 7-3-1, Bunkyo-ku, Tokyo 113-0033 Japan}
\email{ryotaro.niwa@phys.s.u-tokyo.ac.jp}

\author{Satoshi Yoshida}
\affiliation{Department of Physics, Graduate School of Science, The University of Tokyo, Hongo 7-3-1, Bunkyo-ku, Tokyo 113-0033 Japan}

\author{Koki Ono}
\affiliation{Graduate School of Arts and Sciences, The University of Tokyo, 3-8-1 Komaba, Meguro-ku, Tokyo 153-8902, Japan}

\author{Takeru Utsumi}
\affiliation{Graduate School of Arts and Sciences, The University of Tokyo, 3-8-1 Komaba, Meguro-ku, Tokyo 153-8902, Japan}

\author{Zhaoyi Li}
\affiliation{Department of Physics, Massachusetts Institute of Technology, Cambridge MA 02139, USA}

\author{Yuxiang Yang}
\affiliation{Quantum Information and Computation Initiative, Department of Computer Science, School of Computing and Data Science, The University of Hong Kong, Pokfulam Road, Hong Kong, China}

\author{Ryuji Takagi}
\affiliation{Graduate School of Arts and Sciences, The University of Tokyo, 3-8-1 Komaba, Meguro-ku, Tokyo 153-8902, Japan}

\author{Mio Murao}
\affiliation{Department of Physics, Graduate School of Science, The University of Tokyo, Hongo 7-3-1, Bunkyo-ku, Tokyo 113-0033 Japan}

\date{\today}

\begin{abstract}
We consider the problem of purifying noisy qubit unitary channels. Given the ability to apply an unknown qubit unitary channel followed by depolarizing noise, we aim to construct a superchannel that purifies the noisy unitary back to the original unknown unitary.  We first provide numerical evidence that sequential strategies can strictly outperform parallel strategies when the number of channel uses is finite, highlighting the fundamental distinction from state purification. We then provide a concrete $\mathrm{U}(2)$-covariant parallel protocol based on a novel entanglement-assisted quantum error-correcting code that suppresses the first-order noise strength as $O(1/n)$ with $n$ channel uses and show this scaling is asymptotically optimal in the low-noise regime, even when sequential strategies are allowed.
\end{abstract}

\maketitle

\section{Introduction}
The ability to correct errors is crucial for realizing reliable quantum computation at large scales. A promising route towards this goal is fault-tolerant quantum computation (FTQC)~\cite{PhysRevA.57.127, gottesman2014faulttolerantquantumcomputationconstant} with quantum error-correcting codes (QECCs)~\cite{Gottesman1997,PhysRevA.54.1098, PhysRevA.54.4741, Kitaev_2003, Dennis_2002}. In FTQC, one first encodes logical qubits into a QECC, applies unitary gates that may be faulty, and finally performs decoding operations to remove the detrimental effect of environmental noise.  The extensive degrees of freedom in choosing the encoding and decoding operations for QECCs have led to a variety of proposals including concatenated codes~\cite{PhysRevA.54.4741}, topological codes~\cite{Kitaev_2003, PhysRevLett.97.180501}, and entanglement-assisted codes~\cite{Brun_2006}. 

A different but similarly motivated task is \textit{state purification}, where one typically assumes access to multiple copies of a noise-corrupted state $\rho$ and aims to approximately recover a single copy of the original unknown pure state. This assumption creates a non-trivial difference from FTQC, where the input is usually specified as a single input state. While the optimal protocol for qubit state purification has long been known~\cite{PhysRevLett.82.4344}, an analytical solution in larger dimensions remained elusive for decades. Recently, several works have considered state purification in general dimensions~\cite{Childs_2025, grier2025streamingquantumstatepurification, brahmachari2025optimalqubitpurificationunitary, li2025optimalquantumpurityamplification}. In particular, Reference~\cite{li2025optimalquantumpurityamplification} provided an optimal deterministic protocol for state purification in arbitrary dimensions and determined the optimal fidelity expression to leading order. 

Analogously, one can consider \textit{channel purification}, where one has access to noisy unitary channels and aims to obtain the best possible approximation to the original unknown unitary gate. Such a distillation of quantum channels poses a qualitatively different challenge from the state case, since one has the freedom to apply quantum operations in various causal orders~\cite{Chiribella_2008}, giving rise to a much broader family of protocols. Note that channel purification is distinct from virtual protocols~\cite{PRXQuantum.6.020325}, where expectation values of observables can be virtually estimated but purified channels themselves are not deterministically obtained. Recently, a non-trivial channel purification protocol for qubit unitary channels under the depolarizing channel was reported~\cite{zhao2026distillingunitaryoperationsnogo}. However, their analysis was limited to $3$ channel uses, leaving the noise reduction rate in the large $n$ limit elusive. 
\begin{figure}
  \centering
  \includegraphics[width=\linewidth]{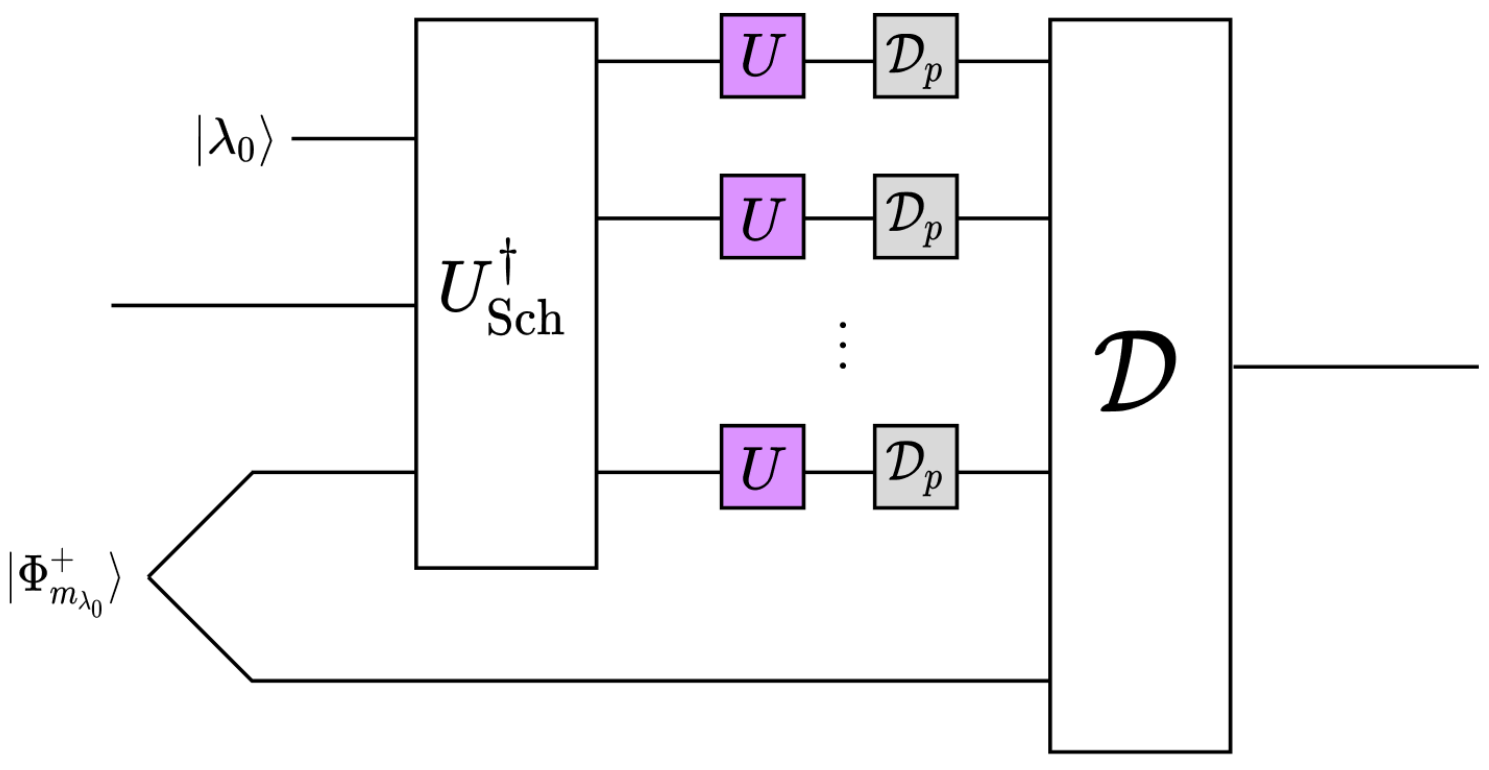}
  \caption{An entanglement-assisted $U(2)$-covariant quantum error-correcting code (QECC) that purifies $n$ noisy channels, which achieves the optimal scaling to first order in the noise strength. For odd integer $n=2k-1$, we define $\lambda_0:=(k, k-1)$, and $|\Phi_{m_{\lambda_0}}^+\rangle :=\frac{1}{\sqrt{m_{\lambda_0}}}\sum_{i=1}^{m_{\lambda_0}}|i\rangle \otimes |i\rangle$. The unitary $U_{\mathrm{Sch}}$ denotes the quantum Schur transform~\cite{bacon2005quantumschurtransformi, Kirby_2018, burchardt2025highdimensionalquantumschurtransforms}, while $\mathcal{D}$ is the optimal recovery map associated to this QECC and noise.}
  \label{fig:setup}
\end{figure}

In this paper, we focus on channel purification for qubit unitary channels under the canonical depolarizing noise. We first show that while the $3$-slot scheme presented in~\cite{zhao2026distillingunitaryoperationsnogo} can be reduced to a parallel scheme, sequential strategies can strictly outperform parallel strategies with $4$ channel uses, showcasing the fundamental distinction from state purification~\cite{li2025optimalquantumpurityamplification}. We then show that the first-order noise strength can be suppressed as $O(\frac{1}{n})$ with $n$ channel uses and prove that this scaling is optimal to the first order. Our upper bound on the fidelity holds for the most general sequential protocols, while our scaling-matching lower bound is provided by a concrete parallel protocol based on a novel $U(2)$-covariant entanglement-assisted QECC. Thus, sequential strategies have no scaling advantage over parallel strategies in the low-noise limit. 

\section{Preliminaries}
Here we review some concepts and tools including quantum superchannels, covariant quantum error-correction, and quantum metrology. The concept of quantum superchannels formalizes our problem setup, while covariant QECCs and quantum metrology are used in our lower bound and upper bound proof, respectively.

\subsection{Quantum superchannels}
A quantum superchannel~\cite{Chiribella_2008, 8678741} is a linear map that transforms quantum channels to quantum channels. Formally, an $n$-slot superchannel $\Xi$ is a linear map
\begin{align}
    \Xi: \bigotimes_{i=1}^{n} [\mathcal{L}(\mathcal{I}_i) \to \mathcal{L}(\mathcal{O}_{i})] \to [\mathcal{L}(\mathcal{P}) \to \mathcal{L}(\mathcal{F})] 
\end{align}
such that, for any set of quantum channels $\Phi_i: \mathcal{L}(\mathcal{I}_i \otimes \mathcal{A}_i) \to \mathcal{L}(\mathcal{O}_{i} \otimes \mathcal{A}_i)$ and auxiliary systems $\mathcal{A}_i$ for $i \in [n]$, the output $\Phi_{\textrm{out}}=(\Xi\otimes \mathbb{I}_{\mathcal{A}_1\ldots \mathcal{A}_n})(\Phi_1\otimes \Phi_2\otimes \cdots \otimes\Phi_n)$ is a quantum channel. Here, $\mathbb{I}$ denotes the identity. In this work, we consider the case where the input quantum channels $\Phi_1, \ldots, \Phi_n$ are identical. In this case, the most general quantum superchannel implementable in a quantum circuit is described as
\begin{align}\label{eq:Seqdecomp}
    \Phi_{\textrm{out}}= \Lambda_{n+1} \circ (\Phi_n \otimes \mathbb{I}_{\mathcal{A}_n}) \circ \Lambda_{n} \circ \cdots \circ (\Phi_{1} \otimes \mathbb{I}_{\mathcal{A}_1}) \circ \Lambda_1 
\end{align}
using auxiliary Hilbert spaces $\mathcal{A}_1, \cdots, \mathcal{A}_n$ and quantum channels $\Lambda_i: \mathcal{L}(\mathcal{O}_{i-1}\otimes \mathcal{A}_{i-1}) \to \mathcal{L}(\mathcal{I}_{i}\otimes \mathcal{A}_{i})$, $i\in [n+1]$, where $\mathcal{O}_0 := \mathcal{P}$ and $\mathcal{I}_{n+1} := \mathcal{F}$. Such a superchannel $\Xi$ is called a \textit{sequential superchannel}, a sequential strategy, or a quantum comb~\cite{chiribella2008quantum}, depending on the context. A restricted class of sequential superchannels is the \textit{parallel superchannel} (or parallel strategies), which admits the decomposition 
\begin{align}
    \Phi_{\textrm{out}} = \mathcal{D} \circ \qty(\bigotimes_{i=1}^n \Phi_i \otimes \mathbb{I}_\mathcal{A}) \circ \mathcal{E}, 
\end{align}
using an encoding channel $\mathcal{E}: \mathcal{L}(\mathcal{P})\to \mathcal{L}(\mathbf{I} \otimes \mathcal{A})$ with $\mathbf{I} = \bigotimes_i \mathcal{I}_{i}$ and the decoding channel $\mathcal{D}:  \mathcal{L}(\mathbf{O} \otimes \mathcal{A}) \to \mathcal{L}(\mathcal{F})$ with $\mathbf{O} = \bigotimes_i \mathcal{O}_{i}$. Analogous to quantum channels, one can define the \textit{Choi matrix} of quantum superchannels~\cite{Chiribella_2008} (See Appendix~\ref{ap:Background}). 


\subsection{Covariant QECC}
Consider a physical system that decomposes into local subsystems: $\mathcal{H}_S=\bigotimes_{i=1}^n \mathcal{H}_{S_{i}}$. A quantum error-correcting code (QECC) can be defined as the image of an encoding channel $\mathcal{E}_{L\to S}: \mathcal{L}(\mathcal{H}_{\textrm{L}}) \to  \mathcal{L}(\mathcal{H}_{\textrm{S}})$.
When the system undergoes a noisy operation $\mathcal{N}_S:  \mathcal{L}(\mathcal{H}_{\textrm{S}}) \to  \mathcal{L}(\mathcal{H}_{\textrm{S}})$, one applies a decoding map $\mathcal{D}_{S\to L}: \mathcal{L}(\mathcal{H}_{\textrm{S}}) \to  \mathcal{L}(\mathcal{H}_{\textrm{L}})$ so that the composed channel 
\begin{align}
     \mathcal{I}_L = \mathcal{D}_{S\to L} \circ \mathcal{N}_S \circ \mathcal{E}_{L\to S}
\end{align}
is as close as possible to the identity channel $\mathbb{I}_L$. When $\mathcal{I}_L = \mathbb{I}_L$, we say that the code is \textit{exact}~\cite{Knill_2000}.
When $\mathcal{I}_L \simeq \mathbb{I}_L$, we say that the code is \textit{approximate}~\cite{Barnum_2002, PhysRevLett.104.120501, PhysRevA.81.062342}.
To quantify the performance of an approximate QECC, we use
\begin{align}
    f_{\textrm{Choi}}(\Phi_1, \Phi_2):= f\qty(\widetilde{\mathcal{J}}_{\Phi_1}, \widetilde{\mathcal{J}}_{\Phi_2}),
\end{align}
where $f(\rho, \sigma):= [\Tr \sqrt{\rho^{1/2} \sigma \rho^{1/2}}]^2$ is the squared fidelity, $\widetilde{\mathcal{J}}_{\Phi_1}:= (\Phi_1\otimes \mathbb{I}_R)(|\Phi_{d}^+\rangle \langle \Phi_{d}^+|)$, $|\Phi_d^+\rangle := \frac{1}{\sqrt{d}}\sum_{i=1}^{d}|i\rangle \otimes |i\rangle$.  The optimal recovery fidelity in the average sense is~\footnote{Note that the references~\cite{Zhou_2021, Kubica_2021, Faist_2020} use the square root fidelity $f_{\mathrm{root}}:=\Tr \sqrt{\rho^{1/2} \sigma \rho^{1/2}}$.} 
\begin{align}
    f_{\mathrm{av}} := \max_{\mathcal{D}} f_{\textrm{Choi}}(\mathcal{D} \circ \mathcal{N} \circ \mathcal{E}, \mathbb{I}_L).
\end{align}
Now, consider any Lie group $G$ acting unitarily on the logical and physical systems with representations $U_L(g), U_S(g) \, (g\in G)$, respectively. We call a code $\mathcal{E}_{L\to S}$ $G$-covariant if 
\begin{align}
    \mathcal{E}_{L\to S} \circ \mathcal{U}_L(g) = \mathcal{U}_S(g) \circ \mathcal{E}_{L\to S}, 
\end{align}
for all $g\in G$, where $\mathcal{U}_L(g)(\cdot) = U_L(g) (\cdot) U_L(g)^\dagger$, $\mathcal{U}_S(g)(\cdot) = U_S(g) (\cdot) U_S(g)^\dagger$.
Previous works have revealed that such a covariance puts fundamental restrictions on the code's performance~\cite{Faist_2020}. 

\subsection{Quantum metrology}
The symmetric logarithmic derivative (SLD) quantum Fisher information (QFI)~\cite{cramerrao1969, Holevo} of a one-parameter family of quantum states $\rho_\theta$ is defined as  
\begin{align}
    F(\rho_\theta) := \Tr(\rho_\theta(L_\theta)^2), 
\end{align}
where $L_\theta$ is implicitly defined by 
\begin{align}
    \dot{\rho}_\theta = \frac{1}{2}(L_\theta \rho_\theta + \rho_\theta L_\theta). 
\end{align}
Using the SLD QFI for quantum states, one can define the SLD QFI of a quantum channel $\mathcal{N}_\theta$~\cite{FujiwaraImai2008FibreBundle} by 
\begin{align}
\max_\rho F((\mathcal{N}_\theta \otimes \mathbb{I}_\mathcal{A})(\rho)),  
\end{align}
where $\mathcal{A}$ is an auxiliary reference system.
It is known that the regularized SLD QFI for quantum channels has a single-letter expression~\cite{Zhou_2021_CE}: For a quantum channel $\mathcal{N}_\theta$ with Kraus representation $\mathcal{N}_\theta = \sum_{i=1}^r K_{i,\theta} (\cdot) K_{i, \theta}^\dagger$, the regularized SLD QFI is represented as
\begin{align}
    &F_{\mathrm{reg}}(\mathcal{N}_\theta) := \lim_{n \to \infty} \frac{F(\mathcal{N}_\theta^{\otimes n})}{n} = \begin{dcases}
        4 \min_{h: \beta_\theta=0} \lVert \alpha_\theta \rVert \quad (\ast)\\
        +\infty \quad (\textrm{otherwise})
    \end{dcases},
\end{align}
where $*$ denotes the condition given by
\begin{align}
    \ast : &i \sum_{i=1}^rK_{i, \theta}^\dagger \partial_\theta K_{i,\theta} \in \mathrm{span}\{K_{i,\theta}^\dagger K_{j, \theta}, \:\forall i, j\}\label{eq:HNKS}\\
    &\begin{dcases}
    \alpha_\theta = (\partial_\theta \mathbf{K}_\theta + ih\mathbf{K}_\theta)^\dagger(\partial_\theta \mathbf{K}_\theta + ih\mathbf{K}_\theta)\\
    \beta_\theta = \mathbf{K}_\theta^\dagger h \mathbf{K}_\theta -i \mathbf{K}_\theta^\dagger \partial_\theta \mathbf{K}_\theta. 
    \end{dcases}
\end{align}
Here $\mathbf{K}_\theta^T := (K_{1,\theta}^T, K_{2,\theta}^T, \cdots K_{r,\theta}^T) \in \mathbf{C}^{d \times rd}$ is a block matrix where $T$ is the transpose, $h$ is a Hermitian operator in $\mathbb{C}^{r\times r}$, and $\lVert \cdot \rVert$ denotes the operator norm. QFIs have been employed to quantify the performance of covariant QECCs~\cite{Kubica_2021, Zhou_2021}.

\section{Setup}
We first describe the optimal channel purification problem in dimension $d$ with $n$ channel uses. This problem was recently solved for the special case of $d=2, n=2,3$~\cite{zhao2026distillingunitaryoperationsnogo}. We are given access to a noisy unitary channel 
\begin{align}
    \mathcal{N}_{\mathcal{U},p} = \mathcal{D}_p\circ \mathcal{U}, 
\end{align}
where $\mathcal{D}_p(\rho) = (1-p)\rho+p\Tr(\rho)\frac{I
}{d}$ is the depolarizing channel with strength $p$ 
and $\mathcal{U}(\cdot) = U (\cdot) U^\dagger$ is an \textit{unknown} unitary channel. Our goal is to construct a superchannel $\Xi$ such that 
\begin{align}
\Xi(\mathcal{N}_{\mathcal{U},p}^{\otimes n}) \simeq \mathcal{U}
\end{align}
for all $U\in \mathrm{U}(d)$.
To quantify the performance of $\Xi$, we use the channel fidelity $f_{\mathrm{Choi}}$ as a figure of merit. It is convenient to define the performance operator~\cite{chiribella2016optimal} for analyzing the optimal performance. We define the performance operator by
\begin{align}
    \Omega_{n,p} = \frac{1}{d^2} \int \mathrm{d} U\, \mathcal{J}_{\mathcal{U}^*} \otimes (\mathcal{J}_{\mathcal{D}_p \circ \mathcal{U}})^{\otimes n}, 
\end{align}
where $(\cdot)^*$ represents the complex conjugate in the computational basis. The average channel fidelity between the output channel $\mathcal{Q}_{\mathcal{U}, p}^{(n)}=\Xi(\mathcal{N}_{\mathcal{U},p}^{\otimes n})$ and the unitary channel $\mathcal{U}$ is given by
\begin{align}
    \int \mathrm{d}U\, f_{\mathrm{Choi}}(\mathcal{Q}_{\mathcal{U}, p}^{(n)}, \mathcal{U}) = \Tr(\Omega_{n,p}^T \mathcal{J}_\Xi). 
\end{align}
Since the performance operator satisfies 
\begin{align}
    [V^*_P \otimes V^{\otimes n}_{\mathbf{I}} \otimes W_{F}^* \otimes W_{\mathbf{O}}^{\otimes n}, \Omega] = 0\quad (\forall V,W\in \mathrm{U}(d)), 
\end{align}
one can twirl $\mathcal{J}_\Xi$ and assume without loss of generality that the Choi matrix satisfies the same symmetry.
Using this twirled protocol, the average fidelity equals the optimal fidelity for every $U$. Thus, the optimal channel purification problem can be formulated as the semidefinite program (SDP)
\begin{align}\label{eq:SDP}
    &\textrm{max} \Tr(\Omega_{n,p}^T \mathcal{J}_\Xi)\nonumber\\
    &\textrm{subject to ``$\Xi$ is a superchannel"}. 
\end{align}
The superchannel constraint is given by Eq.~\eqref{eq:seqChoi}  for sequential superchannels and Eq.~\eqref{eq:parChoi} for parallel superchannels (See Appendix~\ref{ap:Background}). We denote the corresponding optimal values by $f_{n,d}^{\mathrm{Seq}}(p), f_{n,d}^{\mathrm{Par}}(p)$, respectively. Since every parallel superchannel is a special case of sequential superchannels, we have $f_{n,d}^{\mathrm{Par}}(p) \leq f_{n,d}^{\mathrm{Seq}}(p)$. The goal of this paper is to evaluate the large-$n$ scaling of these optimal fidelities in the $p\ll 1$ region, fixing $d=2$. 

\section{Result}
\subsection{Advantage of sequential strategies for finite $n$}
\begin{table*}[t]
\label{tab:n=4}
\centering
\caption{Optimal fidelity obtained numerically for $d=2, n=4, 5$. The upper and lower tables correspond to $n=4,5$, respectively and ``Seq" and ``Par" stand for sequential and parallel strategies. The $n=4$ data indicate that sequential strategies have a strict advantage over parallel strategies, while the $n=5$ data show an apparent coincidence of the optimal fidelities for the sequential and parallel cases. The largest gaps between the sequential and parallel strategies for $n=4,5$ are $5 \times 10^{-3}$ and $2\times 10^{-6}$, respectively. The solver tolerances are of the order $1 \times 10^{-8}$. We note that the performance operator was built from a floating-point group-theoretic data.}
\resizebox{\textwidth}{!}{
\begin{tabular}{c|ccccccccccc}
\hline
$n=4, p$ 
& $0.0$ & $0.1$ & $0.2$ & $0.3$ & $0.4$ & $0.5$ 
& $0.6$ & $0.7$ & $0.8$ & $0.9$ & $1.0$ \\
\hline
Par
& $1.000000$ & $0.966316$ & $0.918284$ & $0.857583$ & $0.786040$ 
& $0.705632$ & $0.618484$ & $0.526865$ & $0.433192$ & $0.340016$ & $0.250000$ \\
Seq
& $1.000000$ & $0.967854$ & $0.921203$ & $0.861569$ & $0.790664$ 
& $0.710400$ & $0.622895$ & $0.530474$ & $0.435671$ & $0.341216$ & $0.250000$ \\
\hline
\end{tabular}
}
\end{table*}
\begin{table*}[t]
\label{tab:n=5}
\centering
\resizebox{\textwidth}{!}{
\begin{tabular}{c|ccccccccccc}
\hline
$n=5, p$
& $0.0$ & $0.1$ & $0.2$ & $0.3$ & $0.4$ & $0.5$
& $0.6$ & $0.7$ & $0.8$ & $0.9$ & $1.0$ \\
\hline
Par
& $1.000000$ & $0.976162$ & $0.940968$ & $0.891548$ & $0.827030$
& $0.748241$ & $0.657405$ & $0.557839$ & $0.453652$ & $0.349444$ & $0.250000$ \\
Seq
& $1.000000$ & $0.976163$ & $0.940970$ & $0.891548$ & $0.827030$
& $0.748241$ & $0.657405$ & $0.557839$ & $0.453652$ & $0.349444$ & $0.250000$ \\
\hline
\end{tabular}
}
\end{table*}
We first solve the SDP in Eq.~\eqref{eq:SDP} numerically for $d=2$ and small $n$~\cite{Grinko_2024, Yoshida_2023}. We use symmetry to reduce the complexity of the problem. First, notice that for $d=2$, 
\begin{align}
    (-i\sigma^y)U^*(-i\sigma^y)^\dagger = (\mathrm{det}U)^{-1}\, U, \quad U\in U(2) 
\end{align}
where $-i\sigma^y=\mqty[0 & -1\\
    1 & 0]$. One can thus rewrite the performance operator without loss of generality as $\Omega \to \widetilde{\Omega}$,   
\begin{align}
     \widetilde{\Omega}_{n,p} := \frac{1}{d^2}\int dU |U\rangle \rangle \langle \langle U| \otimes \mathcal{J}_{\mathcal{D}_p \circ \mathcal{U}}^{\otimes n}.
\end{align}
The operator $\widetilde{\Omega}$ respects the symmetry
\begin{align}\label{eq:PO_unitary_symmetry}
    [V^{\otimes N}_{\mathbf{I}P}\otimes W^{\otimes N}_{\mathbf{O}F}, \widetilde{\Omega}_{n,p}] = 0 \quad \forall V,W \in U(2),  
\end{align}
where $N:=n+1$. Therefore, we may twirl the Choi matrix $\mathcal{J}_\Xi$ and assume without loss of generality that 
\begin{align}\label{eq:CO_unitary_symmetry}
    [V^{\otimes N}_{\mathbf{I}P}\otimes W^{\otimes N}_{\mathbf{O}F}, \mathcal{J}_\Xi] = 0, \quad \forall V,W \in U(2). 
\end{align}
By the Schur-Weyl duality (See Appendix~\ref{ap:Background}), the operator $\mathcal{J}_\Xi$ satisfying Eq.~\eqref{eq:CO_unitary_symmetry} is written as
\begin{align}
\begin{dcases}
    \mathcal{J}_\Xi = \sum_{\mu, \nu \vdash N} \mathbb{I}_{d_\mu} \otimes \mathbb{I}_{d_\nu}  \otimes C_{\mu \nu}\\
     C_{\mu \nu} = \sum_{ijkl}c_{ijkl}^{\mu \nu} |i\rangle \langle j| \otimes |k\rangle \langle l|. 
\end{dcases}
\end{align}
This reduction greatly reduces the size of the optimization~\cite{Grinko_2024,Yoshida_2023}, allowing us to numerically solve the SDP for $d=2, n=3,4,5$. The details of the simplified SDP can be found in Appendix~\ref{ap:SDP}. Building on previous works on unitary inversion~\cite{Quintino_2022, Yoshida_2023}, we computed the simplified SDP in MATLAB~\cite{matlab} using the interpreter CVX~\cite{cvx, gb08} with the solvers SDPT3~\cite{sdpt3, toh1999sdpt3, tutuncu2003solving} and SeDuMi~\cite{sedumi}. Group-theoretic numerical values for computing the simplified SDP were adapted from Reference~\cite{Yoshida_2023}, which originally used SageMath~\cite{sagemath} for computation. 

We first find that the optimal fidelity for $n=3$ reported in Reference~\cite{zhao2026distillingunitaryoperationsnogo} can be achieved by a parallel protocol. Therefore, sequential strategies have no advantage over parallel strategies for $d=2, n=3$. In contrast, the numerical values in TABLE I show for $n=4$ that sequential strategies can strictly outperform parallel strategies. For $n=5$, the sequential and parallel optima seem to agree within numerical precision, suggesting a difference between the even and odd cases. 

\subsection{Lower bound}
We now prove a lower bound on the channel purification fidelity by constructing a concrete $U(2)$-covariant entanglement-assisted QECC. 
\begin{thm}\label{thm1}
The low-noise expansion of the optimal channel purification fidelities $f_{n,2}^{\mathrm{Seq}}(p), f_{n,2}^{\mathrm{Par}}(p)$ at fixed odd integer $n$ satisfies 
\begin{align}
f_{n,2}^{\mathrm{Seq}}(p) 
\geq f_{n,2}^{\mathrm{Par}}(p) \geq 1-C_n^{(1)}p + O(n^2p^2), 
\end{align}
where 
\begin{align}
    C_n^{(1)} := \frac{3}{4}n-\frac{n-1}{3}\qty(1+\frac{1}{2}\sqrt{\frac{n+3}{n}})^2 = \frac{9}{8n} + O\qty(\frac{1}{n^2}). 
\end{align}
This gives the $p\to 0$ expansion for all integers $n$, 
\begin{align}
f_{n,2}^{\mathrm{Seq}}(p) 
\geq f_{n,2}^{\mathrm{Par}}(p) \geq 1- \frac{9p}{8n} + O\qty(\frac{p}{n^2}) + O(n^2p^2). 
\end{align}
\end{thm}
\begin{proof}
Let $n$ be an odd integer, $n=2k-1$ and set $\lambda_0 =(k,k-1)$. We encode the logical state $\ket{\psi}$ into the Weyl module of the irrep, while maximally entangling the Specht module with the noiseless register $E$ (see also Fig.~\ref{fig:setup}). 
\begin{align}
    V_n |\psi\rangle = \frac{1}{\sqrt{m_{\lambda_0}}} \sum_{i=1}^{m_{\lambda_0}} U_{\mathrm{Sch}}^\dagger (|\lambda_0\rangle \otimes |\psi\rangle_{\mathcal{U}_{\lambda_0}} \otimes |i\rangle_{\mathcal{S}_{\lambda_0}}) \otimes |i\rangle_E. 
\end{align}
This encoding is $U(2)$-covariant, that is,  
\begin{align}
    \mathcal{V}_n \circ \mathcal{U} = (\mathcal{U}^{\otimes n}\otimes \mathbb{I}_E) \circ \mathcal{V}_n, 
\end{align}
where $\mathcal{V}_n(\cdot):= V_n(\cdot) V_n^\dagger$. Our proof strategy is to prove a lower bound on the optimal decoding fidelity of this covariant QECC under the depolarizing noise $\mathcal{N}:=\mathcal{D}_p^{\otimes n} \otimes \mathbb{I}_E$. Note that we are considering an entanglement-assisted QECC, that is, the decoding channel is allowed to act on the noiseless system $E$. Now, to the first order of $p$, the $n$-qubit depolarizing channel is 
\begin{align}\label{eq:depolarizingexpansion}
    \mathcal{D}_p^{\otimes n}(\cdot) = q(\cdot) + \epsilon_p \sum_{r=1}^n \sum_{a=x,y,z} \sigma_{a,r}(\cdot) \sigma_{a,r}+ O(n^2p^2), 
\end{align}
where $q:=\qty(1-\frac{3}{4}p)^n$, $\epsilon_p := \frac{p}{4}\qty(1-\frac{3}{4}p)^{n-1}$ and the remainder term is evaluated in the diamond norm. The omitted terms contribute to the fidelity with weight $O(n^2p^2)$, and thus are irrelevant to the first-order coefficient $C_n^{(1)}$. Let us define the Kraus operators of $\mathcal{D}_p^{\otimes n}$ that contribute to the no-error and single-error terms as 
\begin{align}
    E_{0} := \sqrt{q} \: \mathbb{I}, \quad E_{a,r}:= \sqrt{\epsilon_p} \sigma_{a,r}, 
\end{align}
and write $\mathcal{N}_{\mathrm{eff}}$ to denote the no-error and single-error part of Eq.~\eqref{eq:depolarizingexpansion}. Utilizing the Bény-Oreshkov condition~\cite{PhysRevLett.104.120501, Faist_2020}, which states that the fidelity with which one can reverse the action of the encoding and the noise channel is exactly the fidelity of the total complementary channel to a constant channel, the average decoding fidelity $f_{\mathrm{av}}$ may be expressed as 
\begin{align}
f_{\mathrm{av}} &= \max_\rho f_{\mathrm{Choi}}(\widehat{\mathcal{N} \circ \mathcal{E}}, \mathcal{T}_\rho)\nonumber\\
&=\max_\rho f_{\mathrm{Choi}}(\widehat{\mathcal{N}_{\mathrm{eff}} \circ \mathcal{E}}, \mathcal{T}_\rho) + O(n^2p^2), 
\end{align}
where $\widehat{\mathcal{N}}$ denotes the complementary  channel of $\mathcal{N}$ defined by
$\widehat{\mathcal{N}}(\cdot):=\Tr_{\mathcal{O}}[V(\cdot)V^\dagger]$, where $V:\mathcal{I} \to \mathcal{O} \otimes \mathcal{E}$ is the Stinespring dilation of the channel $\mathcal{N}$, $\mathcal{N}(\cdot):=\Tr_{\mathcal{E}}[V(\cdot)V^\dagger]$, and $\mathcal{T}_\rho(\cdot):= \Tr(\cdot)\rho$. Here, $\mathcal{N}_{\mathrm{eff}}$ is a trace-non-increasing completely positive map corresponding to the no-error and single-error sectors. The omitted multi-error sector has total weight $O(n^2p^2)$, and hence changes the optimized squared fidelity only by $O(n^2p^2)$. In our case, $\widehat{\mathcal{N}_\mathrm{eff} \circ \mathcal{V}_n}(\rho)$ has the form 
\begin{align}
\widehat{\mathcal{N}_\mathrm{eff} \circ \mathcal{V}_n}(\rho) = \sum_{ij} \Tr(\rho V_n^\dagger E_j^\dagger E_i V_n) |i\rangle \langle j|, 
\end{align}
\begin{align}\label{eq:compchannelelem}
    V_n^\dagger E_j^\dagger E_i V_n = A_{ij} \mathbb{I}_L + \sum_{c=x,y,z} (\delta A)_{ij}^c \tau_c, 
\end{align}
with matrix elements $A_{ij}, (\delta A)_{ij}$ given in Appendix~\ref{ap:lowerbound}. In the above equation, $\tau_a$ denotes the logical Pauli operator, $\epsilon_{ijk}$ denotes the Levi-Civita epsilon defined by 
\begin{align}
    \epsilon_{ijk}:=
    \begin{dcases}
        +1 \quad (\textrm{even permutation)}\\
        -1 \quad (\textrm{odd permutation)}\\
        0. 
    \end{dcases}
\end{align}
Now, let $A:= \sum_{ij} A_{ij} |i\rangle \langle j|, \, 
(\delta A)^c := \sum_{ij} (\delta A)^{c}_{ij} |i\rangle \langle j|$,  
\begin{align}
    \rho_S := A \otimes \frac{\mathbb{I}_L}{2} + \frac{1}{2} \sum_{c} (\delta A)^c \otimes \tau_c^T.
\end{align}
The average recovery fidelity $f_{\mathrm{av}}$ satisfies 
\begin{align}\label{eq:fidelitylowerbound}
    f_{\mathrm{av}} &=  \max_\rho f_{\mathrm{Choi}}(\widehat{\mathcal{N} \circ \mathcal{E}}, \mathcal{T}_\rho)\nonumber\\
    &\geq f_(\rho_S, \rho_E \otimes \frac{\mathbb{I}_L}{2})+O(n^2p^2)\nonumber\\
    &= 1-C_n^{(1)} p +O(n^2p^2)\nonumber \\
    &=  1- \frac{9p}{8n} + O\qty(\frac{p}{n^2}) + O(n^2p^2), 
\end{align}
for some optimized state $\rho_E$ and 
\begin{align}
    C_n^{(1)} := \frac{3}{4}n - \frac{n-1}{3}\qty(1 + \frac{1}{2}\sqrt{\frac{n+3}{n}})^2. 
\end{align}
See Appendix~\ref{ap:lowerbound} for details of the proof. Since our construction is a parallel superchannel, we have 
\begin{align}
f_{n,2}^{\mathrm{Seq}}(p) 
\geq f_{n,2}^{\mathrm{Par}}(p) \geq 1-C_n^{(1)}p + O(n^2p^2) 
\end{align}
for all odd integers $n$. For even integers $n=2k$, one can discard one slot and run the optimal protocol for $n=2k-1$, which results in the same fidelity lower bound to the leading order. 
\end{proof}
\noindent Notice that the expression for $C_n^{(1)}$ in Thm.~\ref{thm1} exactly matches the first-order coefficient of the optimal $3$-slot fidelity obtained in Reference~\cite{zhao2026distillingunitaryoperationsnogo}, and agrees well with our observed numerical values for $n=5$. 

\subsection{Upper bound}
We now prove a corresponding upper bound based on quantum metrological arguments. 
\begin{thm} The noisy qubit unitary purification fidelity has the following upper bound 
\begin{align}
     f_{n,2}^{\mathrm{Seq}}(p) \leq 1- \frac{3p}{4n} + O\qty(p^2).
\end{align}
\end{thm}
\begin{proof}
Since twirling does not change the target fidelity, we can assume without loss of generality that the superchannel is $U(2)$-covariant. For such covariant channels, we have 
\begin{align}
    \Xi( (\mathcal{D}_p\circ \mathcal{U}_t)^{\otimes n}) = \mathcal{D}_{p'} \circ \mathcal{U}_t, 
\end{align}
where the effective noise parameter $p'$ satisfies 
\begin{align}\label{eq:pprimef}
f_{n,2}^{(\cdot)}(p)=1-\frac{3}{4}p'
\end{align}
(See Appendix~\ref{ap:concatenation}). Now, consider a one-parameter family of channels $\mathcal{N}_{\mathcal{U}_t, p}$ with $U_t = e^{iHt}$, and consider the task of estimating the parameter $t$ given access to $\mathcal{N}_{\mathcal{U}_t, p}$. Let $F_{N}^{\mathrm{Seq}}(\mathcal{N}_{\mathcal{U}_t, p}), F_{N}^{\mathrm{Par}}(\mathcal{N}_{\mathcal{U}_t, p})$ denote the SLD QFI of the channel $\mathcal{N}_{\mathcal{U}_t, p}$ under parallel and sequential strategies with $N$ uses, respectively. The argument in Reference~\cite{Zhou_2021_CE} proves the bound 
\begin{align}
    F_n^{\mathrm{Par}}(\mathcal{N}_{\mathcal{U}_t, p}) &\leq F_n^{\mathrm{Seq}}(\mathcal{N}_{\mathcal{U}_t, p}) \nonumber \\
    &\leq 4 \min_h[n\lVert \alpha \rVert + n(n-1) \lVert \beta \rVert (\lVert \beta \rVert + 2 \sqrt{\lVert \alpha \rVert})]. 
\end{align}
For our choice of $\mathcal{N}_{\mathcal{U}_t,p}$, the Kraus operators satisfy the condition in Eq.~\eqref{eq:HNKS}. Therefore, the regularized channel SLD QFI becomes~\cite{Zhou_2021_CE}
\begin{align}
    F_{\mathrm{reg}}(\mathcal{N}_{\mathcal{U}_t, p}) &:= \lim_{n \to \infty} \frac{F_n^{\mathrm{Par}}
    (\mathcal{N}_{\mathcal{U}_t, p})}{n} \nonumber \\
    &= \lim_{n \to \infty} \frac{F_n^{\mathrm{Seq}}(\mathcal{N}_{{\mathcal{U}_t}, p})}{n} \nonumber \\
    &= (\Delta H)^2 \frac{2(1-p)^2}{p(3-2p)}, 
\end{align}
where $\Delta H$ denotes the difference between the largest and smallest eigenvalue of $H$. Now, consider applying the $n$-slot purification superchannel independently to $m$ blocks of $n$ noisy channels in parallel,
and then performing the optimal parallel estimation strategy on the resulting $m$ purified channels. This
is a valid sequential estimation strategy using $nm$ original noisy channels, so  
\begin{align}
    F_{m}^{\mathrm{Par}}(\mathcal{N}_{\mathcal{U}_t, p'})  = F_{m}^{\mathrm{Par}}(\Xi(\mathcal{N}_{\mathcal{U}_t, p}^{\otimes n})) \leq F_{nm}^{\mathrm{Seq}}(\mathcal{N}_{\mathcal{U}_t,p})
\end{align}
for any sequential superchannel $\Xi$. Dividing by $m$ with a fixed $n$ and taking the limit $m \to \infty$, one obtains 
\begin{align}
      (\Delta H)^2 \frac{2(1-p')^2}{p'(3-2p')} \leq n  (\Delta H)^2 \frac{2(1-p)^2}{p(3-2p)}. 
\end{align}
Expanding this inequality for fixed $n$ and $p\to 0$ gives
\begin{align}
     \frac{p}{n} + O(p^2) \leq p'. 
\end{align}
Substituting Eq.~\eqref{eq:pprimef} gives the desired upper bound. 
\end{proof}

\section{Discussion}
Combining the covariant QECC lower bound with the quantum metrological upper bound, we conclude 
\begin{align}
    1- \frac{9p}{8n} + O\qty(\frac{p}{n^2}) + O(n^2p^2) \leq f_{n,2}^{\mathrm{Seq}}(p) \leq 1- \frac{3p}{4n}  + O(p^2). 
\end{align}
Thus, the first-order coefficient in the $p\to 0$ expansion scales as $\Theta(\frac{1}{n})$, where $n$ is the number of channel uses. Note that we are interested in the order of limits, $\lim_{n\to \infty} \lim_{p\to 0}$. In this asymptotic sense, sequential strategies have no scaling advantage over parallel strategies, even though they can outperform parallel strategies at finite $n$. For odd $n$, we conjecture that sequential strategies have no advantage over parallel strategies, as implied by our numerical results for $n=3,5$. The origin of this even/odd distinction seems to come from the absence versus presence of the $U(2)$-covariant code constructed in our lower bound proof. 

Let us briefly compare the scaling of the first-order noise coefficient with those of existing protocols. First, the generalized W-state encoding introduced in~\cite{Faist_2020} is $U(d)$ covariant, but only achieves the suboptimal scaling $O(\frac{1}{\sqrt{n}})$ for erasure errors. The thermodynamic codes~\cite{Brand_o_2019} achieve the $O(\frac{1}{n})$ scaling~\cite{Zhou_2021} for the depolarizing noise, but are not $\mathrm{SU}(d)$ covariant. Thus, they cannot be used for universal purification of noisy unitary channels. Reference~\cite{Yang_2022} numerically observes $O(\frac{1}{n})$ scaling for depolarizing noise 
using a covariant QECC but falls short of providing an analytical guarantee. Reference~\cite{Kong_2022} establishes $\Theta(\frac{1}{n})$ scaling for erasure errors under the covariant QECC, but does not analyze the depolarizing channel or consider sequential protocols. Finally, concatenating the $3$-slot superchannel protocol in~\cite{zhao2026distillingunitaryoperationsnogo} yields the suboptimal scaling $O(\frac{1}{n^{0.8125}})$ for the first-order noise coefficient, as discussed in Appendix~\ref{ap:concatenation}.

\section{Conclusion}
In this work, we considered the problem of purifying an unknown qubit unitary channel under depolarizing noise.
We first formulated the problem as an SDP and provided numerical evidence that sequential strategies can outperform parallel strategies for a finite number of channel uses. In particular, our data indicate a sequential advantage for $n=4$, while the parallel and sequential optima seem to agree for $n=5$, suggesting a qualitative distinction between even and odd $n$. This demonstrates a strikingly rich structure of the channel purification problem.  We also provided a novel entanglement-assisted $\mathrm{U}(2)$ covariant quantum error-correcting code that allows us to suppress the first-order noise strength linearly with the code size $n$ in the low-noise regime and showed that this noise suppression scaling is optimal even when sequential strategies are allowed. Thus, although sequential strategies can have a strict advantage over parallel strategies for a finite number of channel uses, the asymptotic noise-suppression scaling in the low-noise regime is the same. 

There are many natural directions for future work. First, it would be interesting to determine the exact
leading-order coefficient of the $p\to 0$ expansion. We note that our lower bound coefficient $C_n^{(1)}$ exactly reproduces the analytical result for $n=3$~\cite{zhao2026distillingunitaryoperationsnogo}, and agrees well with our numerical result for $n=5$. This leads us to conjecture that the lower bound is tight to the leading order for all odd $n$. Rigorously deriving the exact leading-order coefficient under the most general sequential strategies remains an interesting direction for future work. Second, it could be interesting to analyze non-unital channels such as the amplitude damping channel in a similar setup. We note that our derivation for the lower bound could be straightforwardly adapted to other Pauli channels such as the dephasing noise (See Appendix~\ref{ap:dephasing}). Finally, it remains an important open problem to decide whether analogous scaling laws hold in arbitrary dimension $d$, and whether dimension-independent purification protocols exist. 
\acknowledgments
We thank Debbie Leung, Isaac Chuang, and Theerapat Tansuwannont for discussions. We thank ChatGPT-5.5 Thinking/Pro for discussions, proofreading, and assistance in writing MATLAB codes used to produce the data in TABLE I. The output was reviewed, examined, and validated by the authors. YY acknowledges the support of the National Natural Science Foundation of China via the Excellent Young Scientists Fund (Hong Kong and Macau) Project 12322516, and the Hong Kong
Research Grant Council (RGC) through the National Natural Science Foundation of China
(NSFC)/Research Grants Council (RGC) Joint Research
Scheme via Project N HKU7107/24.
\bibliography{main}

\providecommand{\noopsort}[1]{}\providecommand{\singleletter}[1]{#1}%
\begin{thebibliography}{54}%
\makeatletter
\providecommand \@ifxundefined [1]{%
 \@ifx{#1\undefined}
}%
\providecommand \@ifnum [1]{%
 \ifnum #1\expandafter \@firstoftwo
 \else \expandafter \@secondoftwo
 \fi
}%
\providecommand \@ifx [1]{%
 \ifx #1\expandafter \@firstoftwo
 \else \expandafter \@secondoftwo
 \fi
}%
\providecommand \natexlab [1]{#1}%
\providecommand \enquote  [1]{``#1''}%
\providecommand \bibnamefont  [1]{#1}%
\providecommand \bibfnamefont [1]{#1}%
\providecommand \citenamefont [1]{#1}%
\providecommand \href@noop [0]{\@secondoftwo}%
\providecommand \href [0]{\begingroup \@sanitize@url \@href}%
\providecommand \@href[1]{\@@startlink{#1}\@@href}%
\providecommand \@@href[1]{\endgroup#1\@@endlink}%
\providecommand \@sanitize@url [0]{\catcode `\\12\catcode `\$12\catcode `\&12\catcode `\#12\catcode `\^12\catcode `\_12\catcode `\%12\relax}%
\providecommand \@@startlink[1]{}%
\providecommand \@@endlink[0]{}%
\providecommand \url  [0]{\begingroup\@sanitize@url \@url }%
\providecommand \@url [1]{\endgroup\@href {#1}{\urlprefix }}%
\providecommand \urlprefix  [0]{URL }%
\providecommand \Eprint [0]{\href }%
\providecommand \doibase [0]{https://doi.org/}%
\providecommand \selectlanguage [0]{\@gobble}%
\providecommand \bibinfo  [0]{\@secondoftwo}%
\providecommand \bibfield  [0]{\@secondoftwo}%
\providecommand \translation [1]{[#1]}%
\providecommand \BibitemOpen [0]{}%
\providecommand \bibitemStop [0]{}%
\providecommand \bibitemNoStop [0]{.\EOS\space}%
\providecommand \EOS [0]{\spacefactor3000\relax}%
\providecommand \BibitemShut  [1]{\csname bibitem#1\endcsname}%
\let\auto@bib@innerbib\@empty
\bibitem [{\citenamefont {Gottesman}(1998)}]{PhysRevA.57.127}%
  \BibitemOpen
  \bibfield  {author} {\bibinfo {author} {\bibfnamefont {D.}~\bibnamefont {Gottesman}},\ }\bibfield  {title} {\bibinfo {title} {Theory of fault-tolerant quantum computation},\ }\href {https://doi.org/10.1103/PhysRevA.57.127} {\bibfield  {journal} {\bibinfo  {journal} {Phys. Rev. A}\ }\textbf {\bibinfo {volume} {57}},\ \bibinfo {pages} {127} (\bibinfo {year} {1998})}\BibitemShut {NoStop}%
\bibitem [{\citenamefont {Gottesman}(2014)}]{gottesman2014faulttolerantquantumcomputationconstant}%
  \BibitemOpen
  \bibfield  {author} {\bibinfo {author} {\bibfnamefont {D.}~\bibnamefont {Gottesman}},\ }\href {https://arxiv.org/abs/1310.2984} {\bibinfo {title} {Fault-tolerant quantum computation with constant overhead}} (\bibinfo {year} {2014}),\ \Eprint {https://arxiv.org/abs/1310.2984} {arXiv:1310.2984 [quant-ph]} \BibitemShut {NoStop}%
\bibitem [{\citenamefont {Gottesman}(1997)}]{Gottesman1997}%
  \BibitemOpen
  \bibfield  {author} {\bibinfo {author} {\bibfnamefont {D.}~\bibnamefont {Gottesman}},\ }\href@noop {} {\bibinfo {title} {{Stabilizer codes and quantum error correction}}} (\bibinfo {year} {1997}),\ \Eprint {https://arxiv.org/abs/quant-ph/9705052} {arXiv:quant-ph/9705052} \BibitemShut {NoStop}%
\bibitem [{\citenamefont {Calderbank}\ and\ \citenamefont {Shor}(1996)}]{PhysRevA.54.1098}%
  \BibitemOpen
  \bibfield  {author} {\bibinfo {author} {\bibfnamefont {A.~R.}\ \bibnamefont {Calderbank}}\ and\ \bibinfo {author} {\bibfnamefont {P.~W.}\ \bibnamefont {Shor}},\ }\bibfield  {title} {\bibinfo {title} {Good quantum error-correcting codes exist},\ }\href {https://doi.org/10.1103/PhysRevA.54.1098} {\bibfield  {journal} {\bibinfo  {journal} {Phys. Rev. A}\ }\textbf {\bibinfo {volume} {54}},\ \bibinfo {pages} {1098} (\bibinfo {year} {1996})}\BibitemShut {NoStop}%
\bibitem [{\citenamefont {Steane}(1996)}]{PhysRevA.54.4741}%
  \BibitemOpen
  \bibfield  {author} {\bibinfo {author} {\bibfnamefont {A.~M.}\ \bibnamefont {Steane}},\ }\bibfield  {title} {\bibinfo {title} {Simple quantum error-correcting codes},\ }\href {https://doi.org/10.1103/PhysRevA.54.4741} {\bibfield  {journal} {\bibinfo  {journal} {Phys. Rev. A}\ }\textbf {\bibinfo {volume} {54}},\ \bibinfo {pages} {4741} (\bibinfo {year} {1996})}\BibitemShut {NoStop}%
\bibitem [{\citenamefont {Kitaev}(2003)}]{Kitaev_2003}%
  \BibitemOpen
  \bibfield  {author} {\bibinfo {author} {\bibfnamefont {A.}~\bibnamefont {Kitaev}},\ }\bibfield  {title} {\bibinfo {title} {Fault-tolerant quantum computation by anyons},\ }\href {https://doi.org/10.1016/s0003-4916(02)00018-0} {\bibfield  {journal} {\bibinfo  {journal} {Annals of Physics}\ }\textbf {\bibinfo {volume} {303}},\ \bibinfo {pages} {2} (\bibinfo {year} {2003})}\BibitemShut {NoStop}%
\bibitem [{\citenamefont {Dennis}\ \emph {et~al.}(2002)\citenamefont {Dennis}, \citenamefont {Kitaev}, \citenamefont {Landahl},\ and\ \citenamefont {Preskill}}]{Dennis_2002}%
  \BibitemOpen
  \bibfield  {author} {\bibinfo {author} {\bibfnamefont {E.}~\bibnamefont {Dennis}}, \bibinfo {author} {\bibfnamefont {A.}~\bibnamefont {Kitaev}}, \bibinfo {author} {\bibfnamefont {A.}~\bibnamefont {Landahl}},\ and\ \bibinfo {author} {\bibfnamefont {J.}~\bibnamefont {Preskill}},\ }\bibfield  {title} {\bibinfo {title} {Topological quantum memory},\ }\href {https://doi.org/10.1063/1.1499754} {\bibfield  {journal} {\bibinfo  {journal} {Journal of Mathematical Physics}\ }\textbf {\bibinfo {volume} {43}},\ \bibinfo {pages} {4452–4505} (\bibinfo {year} {2002})}\BibitemShut {NoStop}%
\bibitem [{\citenamefont {Bombin}\ and\ \citenamefont {Martin-Delgado}(2006)}]{PhysRevLett.97.180501}%
  \BibitemOpen
  \bibfield  {author} {\bibinfo {author} {\bibfnamefont {H.}~\bibnamefont {Bombin}}\ and\ \bibinfo {author} {\bibfnamefont {M.~A.}\ \bibnamefont {Martin-Delgado}},\ }\bibfield  {title} {\bibinfo {title} {Topological quantum distillation},\ }\href {https://doi.org/10.1103/PhysRevLett.97.180501} {\bibfield  {journal} {\bibinfo  {journal} {Phys. Rev. Lett.}\ }\textbf {\bibinfo {volume} {97}},\ \bibinfo {pages} {180501} (\bibinfo {year} {2006})}\BibitemShut {NoStop}%
\bibitem [{\citenamefont {Brun}\ \emph {et~al.}(2006)\citenamefont {Brun}, \citenamefont {Devetak},\ and\ \citenamefont {Hsieh}}]{Brun_2006}%
  \BibitemOpen
  \bibfield  {author} {\bibinfo {author} {\bibfnamefont {T.}~\bibnamefont {Brun}}, \bibinfo {author} {\bibfnamefont {I.}~\bibnamefont {Devetak}},\ and\ \bibinfo {author} {\bibfnamefont {M.-H.}\ \bibnamefont {Hsieh}},\ }\bibfield  {title} {\bibinfo {title} {Correcting quantum errors with entanglement},\ }\href {https://doi.org/10.1126/science.1131563} {\bibfield  {journal} {\bibinfo  {journal} {Science}\ }\textbf {\bibinfo {volume} {314}},\ \bibinfo {pages} {436–439} (\bibinfo {year} {2006})}\BibitemShut {NoStop}%
\bibitem [{\citenamefont {Cirac}\ \emph {et~al.}(1999)\citenamefont {Cirac}, \citenamefont {Ekert},\ and\ \citenamefont {Macchiavello}}]{PhysRevLett.82.4344}%
  \BibitemOpen
  \bibfield  {author} {\bibinfo {author} {\bibfnamefont {J.~I.}\ \bibnamefont {Cirac}}, \bibinfo {author} {\bibfnamefont {A.~K.}\ \bibnamefont {Ekert}},\ and\ \bibinfo {author} {\bibfnamefont {C.}~\bibnamefont {Macchiavello}},\ }\bibfield  {title} {\bibinfo {title} {Optimal purification of single qubits},\ }\href {https://doi.org/10.1103/PhysRevLett.82.4344} {\bibfield  {journal} {\bibinfo  {journal} {Phys. Rev. Lett.}\ }\textbf {\bibinfo {volume} {82}},\ \bibinfo {pages} {4344} (\bibinfo {year} {1999})}\BibitemShut {NoStop}%
\bibitem [{\citenamefont {Childs}\ \emph {et~al.}(2025)\citenamefont {Childs}, \citenamefont {Fu}, \citenamefont {Leung}, \citenamefont {Li}, \citenamefont {Ozols},\ and\ \citenamefont {Vyas}}]{Childs_2025}%
  \BibitemOpen
  \bibfield  {author} {\bibinfo {author} {\bibfnamefont {A.~M.}\ \bibnamefont {Childs}}, \bibinfo {author} {\bibfnamefont {H.}~\bibnamefont {Fu}}, \bibinfo {author} {\bibfnamefont {D.}~\bibnamefont {Leung}}, \bibinfo {author} {\bibfnamefont {Z.}~\bibnamefont {Li}}, \bibinfo {author} {\bibfnamefont {M.}~\bibnamefont {Ozols}},\ and\ \bibinfo {author} {\bibfnamefont {V.}~\bibnamefont {Vyas}},\ }\bibfield  {title} {\bibinfo {title} {Streaming quantum state purification},\ }\href {https://doi.org/10.22331/q-2025-01-21-1603} {\bibfield  {journal} {\bibinfo  {journal} {Quantum}\ }\textbf {\bibinfo {volume} {9}},\ \bibinfo {pages} {1603} (\bibinfo {year} {2025})}\BibitemShut {NoStop}%
\bibitem [{\citenamefont {Grier}\ \emph {et~al.}(2025)\citenamefont {Grier}, \citenamefont {Leung}, \citenamefont {Li}, \citenamefont {Pashayan},\ and\ \citenamefont {Schaeffer}}]{grier2025streamingquantumstatepurification}%
  \BibitemOpen
  \bibfield  {author} {\bibinfo {author} {\bibfnamefont {D.}~\bibnamefont {Grier}}, \bibinfo {author} {\bibfnamefont {D.}~\bibnamefont {Leung}}, \bibinfo {author} {\bibfnamefont {Z.}~\bibnamefont {Li}}, \bibinfo {author} {\bibfnamefont {H.}~\bibnamefont {Pashayan}},\ and\ \bibinfo {author} {\bibfnamefont {L.}~\bibnamefont {Schaeffer}},\ }\href {https://arxiv.org/abs/2503.22644} {\bibinfo {title} {Streaming quantum state purification for general mixed states}} (\bibinfo {year} {2025}),\ \Eprint {https://arxiv.org/abs/2503.22644} {arXiv:2503.22644 [quant-ph]} \BibitemShut {NoStop}%
\bibitem [{\citenamefont {Brahmachari}\ \emph {et~al.}(2025)\citenamefont {Brahmachari}, \citenamefont {Hulse}, \citenamefont {Pfister},\ and\ \citenamefont {Marvian}}]{brahmachari2025optimalqubitpurificationunitary}%
  \BibitemOpen
  \bibfield  {author} {\bibinfo {author} {\bibfnamefont {S.}~\bibnamefont {Brahmachari}}, \bibinfo {author} {\bibfnamefont {A.}~\bibnamefont {Hulse}}, \bibinfo {author} {\bibfnamefont {H.~D.}\ \bibnamefont {Pfister}},\ and\ \bibinfo {author} {\bibfnamefont {I.}~\bibnamefont {Marvian}},\ }\href {https://arxiv.org/abs/2508.05046} {\bibinfo {title} {Optimal qubit purification and unitary schur sampling via random swap tests}} (\bibinfo {year} {2025}),\ \Eprint {https://arxiv.org/abs/2508.05046} {arXiv:2508.05046 [quant-ph]} \BibitemShut {NoStop}%
\bibitem [{\citenamefont {Li}\ \emph {et~al.}(2025)\citenamefont {Li}, \citenamefont {Fu}, \citenamefont {Isogawa}, \citenamefont {Silva},\ and\ \citenamefont {Chuang}}]{li2025optimalquantumpurityamplification}%
  \BibitemOpen
  \bibfield  {author} {\bibinfo {author} {\bibfnamefont {Z.}~\bibnamefont {Li}}, \bibinfo {author} {\bibfnamefont {H.}~\bibnamefont {Fu}}, \bibinfo {author} {\bibfnamefont {T.}~\bibnamefont {Isogawa}}, \bibinfo {author} {\bibfnamefont {C.}~\bibnamefont {Silva}},\ and\ \bibinfo {author} {\bibfnamefont {I.}~\bibnamefont {Chuang}},\ }\href {https://arxiv.org/abs/2409.18167} {\bibinfo {title} {Optimal quantum purity amplification}} (\bibinfo {year} {2025}),\ \Eprint {https://arxiv.org/abs/2409.18167} {arXiv:2409.18167 [quant-ph]} \BibitemShut {NoStop}%
\bibitem [{\citenamefont {Chiribella}\ \emph {et~al.}(2008{\natexlab{a}})\citenamefont {Chiribella}, \citenamefont {D’Ariano},\ and\ \citenamefont {Perinotti}}]{Chiribella_2008}%
  \BibitemOpen
  \bibfield  {author} {\bibinfo {author} {\bibfnamefont {G.}~\bibnamefont {Chiribella}}, \bibinfo {author} {\bibfnamefont {G.~M.}\ \bibnamefont {D’Ariano}},\ and\ \bibinfo {author} {\bibfnamefont {P.}~\bibnamefont {Perinotti}},\ }\bibfield  {title} {\bibinfo {title} {Transforming quantum operations: Quantum supermaps},\ }\href {https://doi.org/10.1209/0295-5075/83/30004} {\bibfield  {journal} {\bibinfo  {journal} {EPL (Europhysics Letters)}\ }\textbf {\bibinfo {volume} {83}},\ \bibinfo {pages} {30004} (\bibinfo {year} {2008}{\natexlab{a}})}\BibitemShut {NoStop}%
\bibitem [{\citenamefont {Liu}\ \emph {et~al.}(2025)\citenamefont {Liu}, \citenamefont {Zhang}, \citenamefont {Fei},\ and\ \citenamefont {Cai}}]{PRXQuantum.6.020325}%
  \BibitemOpen
  \bibfield  {author} {\bibinfo {author} {\bibfnamefont {Z.}~\bibnamefont {Liu}}, \bibinfo {author} {\bibfnamefont {X.}~\bibnamefont {Zhang}}, \bibinfo {author} {\bibfnamefont {Y.-Y.}\ \bibnamefont {Fei}},\ and\ \bibinfo {author} {\bibfnamefont {Z.}~\bibnamefont {Cai}},\ }\bibfield  {title} {\bibinfo {title} {Virtual channel purification},\ }\href {https://doi.org/10.1103/PRXQuantum.6.020325} {\bibfield  {journal} {\bibinfo  {journal} {PRX Quantum}\ }\textbf {\bibinfo {volume} {6}},\ \bibinfo {pages} {020325} (\bibinfo {year} {2025})}\BibitemShut {NoStop}%
\bibitem [{\citenamefont {Zhao}\ \emph {et~al.}(2026)\citenamefont {Zhao}, \citenamefont {Chen}, \citenamefont {Zhen}, \citenamefont {Zhu}, \citenamefont {Chen},\ and\ \citenamefont {Wang}}]{zhao2026distillingunitaryoperationsnogo}%
  \BibitemOpen
  \bibfield  {author} {\bibinfo {author} {\bibfnamefont {J.}~\bibnamefont {Zhao}}, \bibinfo {author} {\bibfnamefont {Y.-A.}\ \bibnamefont {Chen}}, \bibinfo {author} {\bibfnamefont {G.}~\bibnamefont {Zhen}}, \bibinfo {author} {\bibfnamefont {C.}~\bibnamefont {Zhu}}, \bibinfo {author} {\bibfnamefont {R.}~\bibnamefont {Chen}},\ and\ \bibinfo {author} {\bibfnamefont {X.}~\bibnamefont {Wang}},\ }\href {https://arxiv.org/abs/2604.01048} {\bibinfo {title} {Distilling unitary operations: A no-go theorem and minimal realization}} (\bibinfo {year} {2026}),\ \Eprint {https://arxiv.org/abs/2604.01048} {arXiv:2604.01048 [quant-ph]} \BibitemShut {NoStop}%
\bibitem [{\citenamefont {Bacon}\ \emph {et~al.}(2005)\citenamefont {Bacon}, \citenamefont {Chuang},\ and\ \citenamefont {Harrow}}]{bacon2005quantumschurtransformi}%
  \BibitemOpen
  \bibfield  {author} {\bibinfo {author} {\bibfnamefont {D.}~\bibnamefont {Bacon}}, \bibinfo {author} {\bibfnamefont {I.~L.}\ \bibnamefont {Chuang}},\ and\ \bibinfo {author} {\bibfnamefont {A.~W.}\ \bibnamefont {Harrow}},\ }\href {https://arxiv.org/abs/quant-ph/0601001} {\bibinfo {title} {The quantum schur transform: I. efficient qudit circuits}} (\bibinfo {year} {2005}),\ \Eprint {https://arxiv.org/abs/quant-ph/0601001} {arXiv:quant-ph/0601001 [quant-ph]} \BibitemShut {NoStop}%
\bibitem [{\citenamefont {Kirby}\ and\ \citenamefont {Strauch}(2018)}]{Kirby_2018}%
  \BibitemOpen
  \bibfield  {author} {\bibinfo {author} {\bibfnamefont {W.~M.}\ \bibnamefont {Kirby}}\ and\ \bibinfo {author} {\bibfnamefont {F.~W.}\ \bibnamefont {Strauch}},\ }\bibfield  {title} {\bibinfo {title} {A practical quantum algorithm for the schur transform},\ }\href {https://doi.org/10.26421/qic18.9-10-1} {\bibfield  {journal} {\bibinfo  {journal} {Quantum Information and Computation}\ }\textbf {\bibinfo {volume} {18}},\ \bibinfo {pages} {721–742} (\bibinfo {year} {2018})}\BibitemShut {NoStop}%
\bibitem [{\citenamefont {Burchardt}\ \emph {et~al.}(2025)\citenamefont {Burchardt}, \citenamefont {Fei}, \citenamefont {Grinko}, \citenamefont {Larocca}, \citenamefont {Ozols}, \citenamefont {Timmerman},\ and\ \citenamefont {Visnevskyi}}]{burchardt2025highdimensionalquantumschurtransforms}%
  \BibitemOpen
  \bibfield  {author} {\bibinfo {author} {\bibfnamefont {A.}~\bibnamefont {Burchardt}}, \bibinfo {author} {\bibfnamefont {J.}~\bibnamefont {Fei}}, \bibinfo {author} {\bibfnamefont {D.}~\bibnamefont {Grinko}}, \bibinfo {author} {\bibfnamefont {M.}~\bibnamefont {Larocca}}, \bibinfo {author} {\bibfnamefont {M.}~\bibnamefont {Ozols}}, \bibinfo {author} {\bibfnamefont {S.}~\bibnamefont {Timmerman}},\ and\ \bibinfo {author} {\bibfnamefont {V.}~\bibnamefont {Visnevskyi}},\ }\href {https://arxiv.org/abs/2509.22640} {\bibinfo {title} {High-dimensional quantum schur transforms}} (\bibinfo {year} {2025}),\ \Eprint {https://arxiv.org/abs/2509.22640} {arXiv:2509.22640 [quant-ph]} \BibitemShut {NoStop}%
\bibitem [{\citenamefont {Gour}(2019)}]{8678741}%
  \BibitemOpen
  \bibfield  {author} {\bibinfo {author} {\bibfnamefont {G.}~\bibnamefont {Gour}},\ }\bibfield  {title} {\bibinfo {title} {Comparison of quantum channels by superchannels},\ }\href {https://doi.org/10.1109/TIT.2019.2907989} {\bibfield  {journal} {\bibinfo  {journal} {IEEE Transactions on Information Theory}\ }\textbf {\bibinfo {volume} {65}},\ \bibinfo {pages} {5880} (\bibinfo {year} {2019})}\BibitemShut {NoStop}%
\bibitem [{\citenamefont {Chiribella}\ \emph {et~al.}(2008{\natexlab{b}})\citenamefont {Chiribella}, \citenamefont {D'Ariano},\ and\ \citenamefont {Perinotti}}]{chiribella2008quantum}%
  \BibitemOpen
  \bibfield  {author} {\bibinfo {author} {\bibfnamefont {G.}~\bibnamefont {Chiribella}}, \bibinfo {author} {\bibfnamefont {G.~M.}\ \bibnamefont {D'Ariano}},\ and\ \bibinfo {author} {\bibfnamefont {P.}~\bibnamefont {Perinotti}},\ }\bibfield  {title} {\bibinfo {title} {{Quantum Circuit Architecture}},\ }\href {https://doi.org/10.1103/PhysRevLett.101.060401} {\bibfield  {journal} {\bibinfo  {journal} {Phys. Rev. Lett.}\ }\textbf {\bibinfo {volume} {101}},\ \bibinfo {pages} {060401} (\bibinfo {year} {2008}{\natexlab{b}})},\ \Eprint {https://arxiv.org/abs/0712.1325} {arXiv:0712.1325} \BibitemShut {NoStop}%
\bibitem [{\citenamefont {Knill}\ \emph {et~al.}(2000)\citenamefont {Knill}, \citenamefont {Laflamme},\ and\ \citenamefont {Viola}}]{Knill_2000}%
  \BibitemOpen
  \bibfield  {author} {\bibinfo {author} {\bibfnamefont {E.}~\bibnamefont {Knill}}, \bibinfo {author} {\bibfnamefont {R.}~\bibnamefont {Laflamme}},\ and\ \bibinfo {author} {\bibfnamefont {L.}~\bibnamefont {Viola}},\ }\bibfield  {title} {\bibinfo {title} {Theory of quantum error correction for general noise},\ }\href {https://doi.org/10.1103/physrevlett.84.2525} {\bibfield  {journal} {\bibinfo  {journal} {Phys. Rev. Lett.}\ }\textbf {\bibinfo {volume} {84}},\ \bibinfo {pages} {2525} (\bibinfo {year} {2000})}\BibitemShut {NoStop}%
\bibitem [{\citenamefont {Barnum}\ and\ \citenamefont {Knill}(2002)}]{Barnum_2002}%
  \BibitemOpen
  \bibfield  {author} {\bibinfo {author} {\bibfnamefont {H.}~\bibnamefont {Barnum}}\ and\ \bibinfo {author} {\bibfnamefont {E.}~\bibnamefont {Knill}},\ }\bibfield  {title} {\bibinfo {title} {Reversing quantum dynamics with near-optimal quantum and classical fidelity},\ }\href {https://doi.org/10.1063/1.1459754} {\bibfield  {journal} {\bibinfo  {journal} {Journal of Mathematical Physics}\ }\textbf {\bibinfo {volume} {43}},\ \bibinfo {pages} {2097} (\bibinfo {year} {2002})}\BibitemShut {NoStop}%
\bibitem [{\citenamefont {B\'eny}\ and\ \citenamefont {Oreshkov}(2010)}]{PhysRevLett.104.120501}%
  \BibitemOpen
  \bibfield  {author} {\bibinfo {author} {\bibfnamefont {C.}~\bibnamefont {B\'eny}}\ and\ \bibinfo {author} {\bibfnamefont {O.}~\bibnamefont {Oreshkov}},\ }\bibfield  {title} {\bibinfo {title} {General conditions for approximate quantum error correction and near-optimal recovery channels},\ }\href {https://doi.org/10.1103/PhysRevLett.104.120501} {\bibfield  {journal} {\bibinfo  {journal} {Phys. Rev. Lett.}\ }\textbf {\bibinfo {volume} {104}},\ \bibinfo {pages} {120501} (\bibinfo {year} {2010})}\BibitemShut {NoStop}%
\bibitem [{\citenamefont {Ng}\ and\ \citenamefont {Mandayam}(2010)}]{PhysRevA.81.062342}%
  \BibitemOpen
  \bibfield  {author} {\bibinfo {author} {\bibfnamefont {H.~K.}\ \bibnamefont {Ng}}\ and\ \bibinfo {author} {\bibfnamefont {P.}~\bibnamefont {Mandayam}},\ }\bibfield  {title} {\bibinfo {title} {Simple approach to approximate quantum error correction based on the transpose channel},\ }\href {https://doi.org/10.1103/PhysRevA.81.062342} {\bibfield  {journal} {\bibinfo  {journal} {Phys. Rev. A}\ }\textbf {\bibinfo {volume} {81}},\ \bibinfo {pages} {062342} (\bibinfo {year} {2010})}\BibitemShut {NoStop}%
\bibitem [{Note1()}]{Note1}%
  \BibitemOpen
  \bibinfo {note} {Note that the references~\cite {Zhou_2021, Kubica_2021, Faist_2020} use the square root fidelity $f_{\protect \mathrm {root}}:=\Tr \protect \sqrt {\rho ^{1/2} \sigma \rho ^{1/2}}$.}\BibitemShut {Stop}%
\bibitem [{\citenamefont {Faist}\ \emph {et~al.}(2020)\citenamefont {Faist}, \citenamefont {Nezami}, \citenamefont {Albert}, \citenamefont {Salton}, \citenamefont {Pastawski}, \citenamefont {Hayden},\ and\ \citenamefont {Preskill}}]{Faist_2020}%
  \BibitemOpen
  \bibfield  {author} {\bibinfo {author} {\bibfnamefont {P.}~\bibnamefont {Faist}}, \bibinfo {author} {\bibfnamefont {S.}~\bibnamefont {Nezami}}, \bibinfo {author} {\bibfnamefont {V.~V.}\ \bibnamefont {Albert}}, \bibinfo {author} {\bibfnamefont {G.}~\bibnamefont {Salton}}, \bibinfo {author} {\bibfnamefont {F.}~\bibnamefont {Pastawski}}, \bibinfo {author} {\bibfnamefont {P.}~\bibnamefont {Hayden}},\ and\ \bibinfo {author} {\bibfnamefont {J.}~\bibnamefont {Preskill}},\ }\bibfield  {title} {\bibinfo {title} {Continuous symmetries and approximate quantum error correction},\ }\href {https://doi.org/10.1103/physrevx.10.041018} {\bibfield  {journal} {\bibinfo  {journal} {Phys. Rev. X}\ }\textbf {\bibinfo {volume} {10}},\ \bibinfo {pages} {041018} (\bibinfo {year} {2020})}\BibitemShut {NoStop}%
\bibitem [{\citenamefont {Helstrom}(1969)}]{cramerrao1969}%
  \BibitemOpen
  \bibfield  {author} {\bibinfo {author} {\bibfnamefont {C.~W.}\ \bibnamefont {Helstrom}},\ }\bibfield  {title} {\bibinfo {title} {Quantum detection and estimation theory},\ }\href {https://doi.org/10.1007/BF01007479} {\bibfield  {journal} {\bibinfo  {journal} {Journal of Statistical Physics}\ }\textbf {\bibinfo {volume} {1}},\ \bibinfo {pages} {231} (\bibinfo {year} {1969})}\BibitemShut {NoStop}%
\bibitem [{\citenamefont {A.S.Holevo}(1982)}]{Holevo}%
  \BibitemOpen
  \bibfield  {author} {\bibinfo {author} {\bibnamefont {A.S.Holevo}},\ }\href@noop {} {\emph {\bibinfo {title} {Probabilistic and Statistical Aspects of Quantum Theory}}}\ (\bibinfo  {publisher} {Springer Science $\&$ Business Media},\ \bibinfo {year} {1982})\BibitemShut {NoStop}%
\bibitem [{\citenamefont {Fujiwara}\ and\ \citenamefont {Imai}(2008)}]{FujiwaraImai2008FibreBundle}%
  \BibitemOpen
  \bibfield  {author} {\bibinfo {author} {\bibfnamefont {A.}~\bibnamefont {Fujiwara}}\ and\ \bibinfo {author} {\bibfnamefont {H.}~\bibnamefont {Imai}},\ }\bibfield  {title} {\bibinfo {title} {A fibre bundle over manifolds of quantum channels and its application to quantum statistics},\ }\href {https://doi.org/10.1088/1751-8113/41/25/255304} {\bibfield  {journal} {\bibinfo  {journal} {Journal of Physics A: Mathematical and Theoretical}\ }\textbf {\bibinfo {volume} {41}},\ \bibinfo {pages} {255304} (\bibinfo {year} {2008})}\BibitemShut {NoStop}%
\bibitem [{\citenamefont {Zhou}\ and\ \citenamefont {Jiang}(2021)}]{Zhou_2021_CE}%
  \BibitemOpen
  \bibfield  {author} {\bibinfo {author} {\bibfnamefont {S.}~\bibnamefont {Zhou}}\ and\ \bibinfo {author} {\bibfnamefont {L.}~\bibnamefont {Jiang}},\ }\bibfield  {title} {\bibinfo {title} {Asymptotic theory of quantum channel estimation},\ }\href {https://doi.org/10.1103/prxquantum.2.010343} {\bibfield  {journal} {\bibinfo  {journal} {PRX Quantum}\ }\textbf {\bibinfo {volume} {2}},\ \bibinfo {pages} {010343} (\bibinfo {year} {2021})}\BibitemShut {NoStop}%
\bibitem [{\citenamefont {Kubica}\ and\ \citenamefont {Demkowicz-Dobrzański}(2021)}]{Kubica_2021}%
  \BibitemOpen
  \bibfield  {author} {\bibinfo {author} {\bibfnamefont {A.}~\bibnamefont {Kubica}}\ and\ \bibinfo {author} {\bibfnamefont {R.}~\bibnamefont {Demkowicz-Dobrzański}},\ }\bibfield  {title} {\bibinfo {title} {Using quantum metrological bounds in quantum error correction: A simple proof of the approximate eastin-knill theorem},\ }\href {https://doi.org/10.1103/physrevlett.126.150503} {\bibfield  {journal} {\bibinfo  {journal} {Phys. Rev. Lett.}\ }\textbf {\bibinfo {volume} {126}},\ \bibinfo {pages} {150503} (\bibinfo {year} {2021})}\BibitemShut {NoStop}%
\bibitem [{\citenamefont {Zhou}\ \emph {et~al.}(2021)\citenamefont {Zhou}, \citenamefont {Liu},\ and\ \citenamefont {Jiang}}]{Zhou_2021}%
  \BibitemOpen
  \bibfield  {author} {\bibinfo {author} {\bibfnamefont {S.}~\bibnamefont {Zhou}}, \bibinfo {author} {\bibfnamefont {Z.-W.}\ \bibnamefont {Liu}},\ and\ \bibinfo {author} {\bibfnamefont {L.}~\bibnamefont {Jiang}},\ }\bibfield  {title} {\bibinfo {title} {New perspectives on covariant quantum error correction},\ }\href {https://doi.org/10.22331/q-2021-08-09-521} {\bibfield  {journal} {\bibinfo  {journal} {Quantum}\ }\textbf {\bibinfo {volume} {5}},\ \bibinfo {pages} {521} (\bibinfo {year} {2021})}\BibitemShut {NoStop}%
\bibitem [{\citenamefont {Chiribella}\ and\ \citenamefont {Ebler}(2016)}]{chiribella2016optimal}%
  \BibitemOpen
  \bibfield  {author} {\bibinfo {author} {\bibfnamefont {G.}~\bibnamefont {Chiribella}}\ and\ \bibinfo {author} {\bibfnamefont {D.}~\bibnamefont {Ebler}},\ }\bibfield  {title} {\bibinfo {title} {Optimal quantum networks and one-shot entropies},\ }\href@noop {} {\bibfield  {journal} {\bibinfo  {journal} {New Journal of Physics}\ }\textbf {\bibinfo {volume} {18}},\ \bibinfo {pages} {093053} (\bibinfo {year} {2016})}\BibitemShut {NoStop}%
\bibitem [{\citenamefont {Grinko}\ and\ \citenamefont {Ozols}(2024)}]{Grinko_2024}%
  \BibitemOpen
  \bibfield  {author} {\bibinfo {author} {\bibfnamefont {D.}~\bibnamefont {Grinko}}\ and\ \bibinfo {author} {\bibfnamefont {M.}~\bibnamefont {Ozols}},\ }\bibfield  {title} {\bibinfo {title} {Linear programming with unitary-equivariant constraints},\ }\href {https://doi.org/10.1007/s00220-024-05108-1} {\bibfield  {journal} {\bibinfo  {journal} {Communications in Mathematical Physics}\ }\textbf {\bibinfo {volume} {405}},\ \bibinfo {pages} {278} (\bibinfo {year} {2024})}\BibitemShut {NoStop}%
\bibitem [{\citenamefont {Yoshida}\ \emph {et~al.}(2023)\citenamefont {Yoshida}, \citenamefont {Soeda},\ and\ \citenamefont {Murao}}]{Yoshida_2023}%
  \BibitemOpen
  \bibfield  {author} {\bibinfo {author} {\bibfnamefont {S.}~\bibnamefont {Yoshida}}, \bibinfo {author} {\bibfnamefont {A.}~\bibnamefont {Soeda}},\ and\ \bibinfo {author} {\bibfnamefont {M.}~\bibnamefont {Murao}},\ }\bibfield  {title} {\bibinfo {title} {Reversing unknown qubit-unitary operation, deterministically and exactly},\ }\href {https://doi.org/10.1103/physrevlett.131.120602} {\bibfield  {journal} {\bibinfo  {journal} {Phys. Rev. Lett.}\ }\textbf {\bibinfo {volume} {131}},\ \bibinfo {pages} {120602} (\bibinfo {year} {2023})}\BibitemShut {NoStop}%
\bibitem [{\citenamefont {Quintino}\ and\ \citenamefont {Ebler}(2022)}]{Quintino_2022}%
  \BibitemOpen
  \bibfield  {author} {\bibinfo {author} {\bibfnamefont {M.~T.}\ \bibnamefont {Quintino}}\ and\ \bibinfo {author} {\bibfnamefont {D.}~\bibnamefont {Ebler}},\ }\bibfield  {title} {\bibinfo {title} {Deterministic transformations between unitary operations: Exponential advantage with adaptive quantum circuits and the power of indefinite causality},\ }\href {https://doi.org/10.22331/q-2022-03-31-679} {\bibfield  {journal} {\bibinfo  {journal} {Quantum}\ }\textbf {\bibinfo {volume} {6}},\ \bibinfo {pages} {679} (\bibinfo {year} {2022})}\BibitemShut {NoStop}%
\bibitem [{\citenamefont {MATLAB}(2021)}]{matlab}%
  \BibitemOpen
  \bibfield  {author} {\bibinfo {author} {\bibnamefont {MATLAB}},\ }\href@noop {} {\emph {\bibinfo {title} {version 9.11.0 (R2021b)}}}\ (\bibinfo  {publisher} {The MathWorks Inc.},\ \bibinfo {address} {Natick, Massachusetts},\ \bibinfo {year} {2021})\BibitemShut {NoStop}%
\bibitem [{\citenamefont {Grant}\ and\ \citenamefont {Boyd}(2020)}]{cvx}%
  \BibitemOpen
  \bibfield  {author} {\bibinfo {author} {\bibfnamefont {M.}~\bibnamefont {Grant}}\ and\ \bibinfo {author} {\bibfnamefont {S.}~\bibnamefont {Boyd}},\ }\href@noop {} {\bibinfo {title} {{CVX}: Matlab software for disciplined convex programming, version 2.2}},\ \bibinfo {howpublished} {\url{http://cvxr.com/cvx}} (\bibinfo {year} {2020})\BibitemShut {NoStop}%
\bibitem [{\citenamefont {Grant}\ and\ \citenamefont {Boyd}(2008)}]{gb08}%
  \BibitemOpen
  \bibfield  {author} {\bibinfo {author} {\bibfnamefont {M.}~\bibnamefont {Grant}}\ and\ \bibinfo {author} {\bibfnamefont {S.}~\bibnamefont {Boyd}},\ }\bibfield  {title} {\bibinfo {title} {Graph implementations for nonsmooth convex programs},\ }in\ \href@noop {} {\emph {\bibinfo {booktitle} {Recent Advances in Learning and Control}}},\ \bibinfo {series and number} {Lecture Notes in Control and Information Sciences},\ \bibinfo {editor} {edited by\ \bibinfo {editor} {\bibfnamefont {V.}~\bibnamefont {Blondel}}, \bibinfo {editor} {\bibfnamefont {S.}~\bibnamefont {Boyd}},\ and\ \bibinfo {editor} {\bibfnamefont {H.}~\bibnamefont {Kimura}}}\ (\bibinfo  {publisher} {Springer-Verlag Limited},\ \bibinfo {year} {2008})\ pp.\ \bibinfo {pages} {95--110},\ \bibinfo {note} {\url{http://stanford.edu/~boyd/graph_dcp.html}}\BibitemShut {NoStop}%
\bibitem [{sdp()}]{sdpt3}%
  \BibitemOpen
  \href@noop {} {}\bibinfo {howpublished} {\url{http://www.math.nus.edu.sg/.mattohkc/sdpt3.html}}\BibitemShut {NoStop}%
\bibitem [{\citenamefont {Toh}\ \emph {et~al.}(1999)\citenamefont {Toh}, \citenamefont {Todd},\ and\ \citenamefont {T{\"u}t{\"u}nc{\"u}}}]{toh1999sdpt3}%
  \BibitemOpen
  \bibfield  {author} {\bibinfo {author} {\bibfnamefont {K.-C.}\ \bibnamefont {Toh}}, \bibinfo {author} {\bibfnamefont {M.~J.}\ \bibnamefont {Todd}},\ and\ \bibinfo {author} {\bibfnamefont {R.~H.}\ \bibnamefont {T{\"u}t{\"u}nc{\"u}}},\ }\bibfield  {title} {\bibinfo {title} {Sdpt3—a matlab software package for semidefinite programming, version 1.3},\ }\href {https://doi.org/10.1080/10556789908805762} {\bibfield  {journal} {\bibinfo  {journal} {Optim. Methods Software}\ }\textbf {\bibinfo {volume} {11}},\ \bibinfo {pages} {545} (\bibinfo {year} {1999})}\BibitemShut {NoStop}%
\bibitem [{\citenamefont {T{\"u}t{\"u}nc{\"u}}\ \emph {et~al.}(2003)\citenamefont {T{\"u}t{\"u}nc{\"u}}, \citenamefont {Toh},\ and\ \citenamefont {Todd}}]{tutuncu2003solving}%
  \BibitemOpen
  \bibfield  {author} {\bibinfo {author} {\bibfnamefont {R.~H.}\ \bibnamefont {T{\"u}t{\"u}nc{\"u}}}, \bibinfo {author} {\bibfnamefont {K.-C.}\ \bibnamefont {Toh}},\ and\ \bibinfo {author} {\bibfnamefont {M.~J.}\ \bibnamefont {Todd}},\ }\bibfield  {title} {\bibinfo {title} {Solving semidefinite-quadratic-linear programs using sdpt3},\ }\href {https://doi.org/10.1007/s10107-002-0347-5} {\bibfield  {journal} {\bibinfo  {journal} {Math. Program.}\ }\textbf {\bibinfo {volume} {95}},\ \bibinfo {pages} {189} (\bibinfo {year} {2003})}\BibitemShut {NoStop}%
\bibitem [{\citenamefont {Sturm}(1999)}]{sedumi}%
  \BibitemOpen
  \bibfield  {author} {\bibinfo {author} {\bibfnamefont {J.~F.}\ \bibnamefont {Sturm}},\ }\bibfield  {title} {\bibinfo {title} {Using sedumi 1.02, a matlab toolbox for optimization over symmetric cones},\ }\href {https://doi.org/10.1080/10556789908805766} {\bibfield  {journal} {\bibinfo  {journal} {Optim. Methods Software}\ }\textbf {\bibinfo {volume} {11}},\ \bibinfo {pages} {625} (\bibinfo {year} {1999})}\BibitemShut {NoStop}%
\bibitem [{\citenamefont {{The Sage Developers}}(2022)}]{sagemath}%
  \BibitemOpen
  \bibfield  {author} {\bibinfo {author} {\bibnamefont {{The Sage Developers}}},\ }\href@noop {} {\emph {\bibinfo {title} {{S}ageMath, the {S}age {M}athematics {S}oftware {S}ystem ({V}ersion 9.7)}}} (\bibinfo {year} {2022}),\ \bibinfo {note} {\url{https://www.sagemath.org}}\BibitemShut {NoStop}%
\bibitem [{\citenamefont {Brandão}\ \emph {et~al.}(2019)\citenamefont {Brandão}, \citenamefont {Crosson}, \citenamefont {Şahinoğlu},\ and\ \citenamefont {Bowen}}]{Brand_o_2019}%
  \BibitemOpen
  \bibfield  {author} {\bibinfo {author} {\bibfnamefont {F.~G. S.~L.}\ \bibnamefont {Brandão}}, \bibinfo {author} {\bibfnamefont {E.}~\bibnamefont {Crosson}}, \bibinfo {author} {\bibfnamefont {M.~B.}\ \bibnamefont {Şahinoğlu}},\ and\ \bibinfo {author} {\bibfnamefont {J.}~\bibnamefont {Bowen}},\ }\bibfield  {title} {\bibinfo {title} {Quantum error correcting codes in eigenstates of translation-invariant spin chains},\ }\href {https://doi.org/10.1103/physrevlett.123.110502} {\bibfield  {journal} {\bibinfo  {journal} {Phys. Rev. Lett.}\ }\textbf {\bibinfo {volume} {123}},\ \bibinfo {pages} {110502} (\bibinfo {year} {2019})}\BibitemShut {NoStop}%
\bibitem [{\citenamefont {Yang}\ \emph {et~al.}(2022)\citenamefont {Yang}, \citenamefont {Mo}, \citenamefont {Renes}, \citenamefont {Chiribella},\ and\ \citenamefont {Woods}}]{Yang_2022}%
  \BibitemOpen
  \bibfield  {author} {\bibinfo {author} {\bibfnamefont {Y.}~\bibnamefont {Yang}}, \bibinfo {author} {\bibfnamefont {Y.}~\bibnamefont {Mo}}, \bibinfo {author} {\bibfnamefont {J.~M.}\ \bibnamefont {Renes}}, \bibinfo {author} {\bibfnamefont {G.}~\bibnamefont {Chiribella}},\ and\ \bibinfo {author} {\bibfnamefont {M.~P.}\ \bibnamefont {Woods}},\ }\bibfield  {title} {\bibinfo {title} {Optimal universal quantum error correction via bounded reference frames},\ }\href {https://doi.org/10.1103/physrevresearch.4.023107} {\bibfield  {journal} {\bibinfo  {journal} {Phys. Rev. Res.}\ }\textbf {\bibinfo {volume} {4}},\ \bibinfo {pages} {023107} (\bibinfo {year} {2022})}\BibitemShut {NoStop}%
\bibitem [{\citenamefont {Kong}\ and\ \citenamefont {Liu}(2022)}]{Kong_2022}%
  \BibitemOpen
  \bibfield  {author} {\bibinfo {author} {\bibfnamefont {L.}~\bibnamefont {Kong}}\ and\ \bibinfo {author} {\bibfnamefont {Z.-W.}\ \bibnamefont {Liu}},\ }\bibfield  {title} {\bibinfo {title} {Near-optimal covariant quantum error-correcting codes from random unitaries with symmetries},\ }\href {https://doi.org/10.1103/prxquantum.3.020314} {\bibfield  {journal} {\bibinfo  {journal} {PRX Quantum}\ }\textbf {\bibinfo {volume} {3}},\ \bibinfo {pages} {020314} (\bibinfo {year} {2022})}\BibitemShut {NoStop}%
\bibitem [{\citenamefont {Choi}(1975)}]{CHOI1975285}%
  \BibitemOpen
  \bibfield  {author} {\bibinfo {author} {\bibfnamefont {M.-D.}\ \bibnamefont {Choi}},\ }\bibfield  {title} {\bibinfo {title} {{Completely positive linear maps on complex matrices}},\ }\href {https://doi.org/10.1016/0024-3795(75)90075-0} {\bibfield  {journal} {\bibinfo  {journal} {Linear Algebra Appl.}\ }\textbf {\bibinfo {volume} {10}},\ \bibinfo {pages} {285} (\bibinfo {year} {1975})}\BibitemShut {NoStop}%
\bibitem [{\citenamefont {Jamio{\l}kowski}(1972)}]{JAMIOLKOWSKI1972275}%
  \BibitemOpen
  \bibfield  {author} {\bibinfo {author} {\bibfnamefont {A.}~\bibnamefont {Jamio{\l}kowski}},\ }\bibfield  {title} {\bibinfo {title} {{Linear transformations which preserve trace and positive semidefiniteness of operators}},\ }\href {https://doi.org/10.1016/0034-4877(72)90011-0} {\bibfield  {journal} {\bibinfo  {journal} {Rep. Math. Phys.}\ }\textbf {\bibinfo {volume} {3}},\ \bibinfo {pages} {275} (\bibinfo {year} {1972})}\BibitemShut {NoStop}%
\bibitem [{\citenamefont {James}(2006)}]{james2006representation}%
  \BibitemOpen
  \bibfield  {author} {\bibinfo {author} {\bibfnamefont {G.~D.}\ \bibnamefont {James}},\ }\href {https://doi.org/10.1007/BFb0067708} {\emph {\bibinfo {title} {The Representation Theory of the Symmetric Groups}}},\ Vol.\ \bibinfo {volume} {682}\ (\bibinfo  {publisher} {Springer},\ \bibinfo {year} {2006})\BibitemShut {NoStop}%
\bibitem [{\citenamefont {Goodman}\ and\ \citenamefont {Wallach}(2009)}]{goodman2009symmetry}%
  \BibitemOpen
  \bibfield  {author} {\bibinfo {author} {\bibfnamefont {R.}~\bibnamefont {Goodman}}\ and\ \bibinfo {author} {\bibfnamefont {N.~R.}\ \bibnamefont {Wallach}},\ }\href {https://doi.org/10.1007/978-0-387-79852-3} {\emph {\bibinfo {title} {Symmetry, Representations, and Invariants}}},\ Vol.\ \bibinfo {volume} {255}\ (\bibinfo  {publisher} {Springer},\ \bibinfo {year} {2009})\BibitemShut {NoStop}%
\bibitem [{SDP()}]{SDPcode}%
  \BibitemOpen
  \href@noop {} {}\bibinfo {note} {\url{https://github.com/rn8128/Qubit-unitary-purification}}\BibitemShut {NoStop}%
\end{thebibliography}%
\clearpage

\appendix
\onecolumngrid
\section{Notation and Background}\label{ap:Background}
\subsection{Notation}
Here we clarify the big-$O$ notation $O(\cdot)$, $\Omega(\cdot)$, and $\Theta(\cdot)$, used in the main text. They are defined as follows: 
\begin{align}
    &f(x) = O(g(x)) \Leftrightarrow \limsup_{x \to \infty} \left| \frac{f(x)}{g(x)} \right| < \infty\\
    &f(x) = \Omega(g(x)) \Leftrightarrow g(x) = O(f(x))\\
    &f(x) = \Theta(g(x)) \Leftrightarrow f(x) = O(g(x)), \textrm{and } f(x) = \Omega(g(x)). 
\end{align}
We also use $O_n(\cdot)$, which means 
\begin{align}
    f_n(x) = O_n(g(x))  \Leftrightarrow \limsup_{x \to 0} \left| \frac{f_n(x)}{g(x)} \right| < \infty
\end{align}
for every fixed $n$. 

\subsection{Choi--Jamio\l{}kowski isomorphism for quantum channels and superchannels }
The Choi--Jamio\l{}kowski isomorphism~\cite{CHOI1975285, JAMIOLKOWSKI1972275} is a convenient way to represent quantum channels and superchannels. Given a quantum channel $\Phi: \mathcal{L}(\mathcal{I}) \to \mathcal{L}(\mathcal{O})$, its Choi matrix is defined by 
\begin{align}
  J_\Phi:=\sum_{i,j=1}^{d_{\mathcal{I}}} \ketbra{i}{j}_{\mathcal{I}} \otimes \Phi(\ketbra{i}{j})_{\mathcal{O}} \in \mathcal{L}(\mathcal{I} \otimes \mathcal{O}), 
\end{align}
satisfying the condition $J_\Phi \geq 0, \Tr_\mathcal{O}(J_\Phi) = \mathbb{I}_\mathcal{I}$. Under the Choi--Jamio\l{}kowski isomorphism, the composition of quantum channels $\Phi_1: \mathcal{L}(\mathcal{X}) \to \mathcal{L}(\mathcal{Y})$ and $\Phi_2: \mathcal{L}(\mathcal{Y}) \to \mathcal{L}(\mathcal{Z})$ is represented by the link product~\cite{chiribella2008quantum} defined by
\begin{align}
  J_{\Phi_2 \circ \Phi_1}
  &= J_{\Phi_2} \star J_{\Phi_1}\\
  &\coloneqq \Tr_{\mathcal{Y}}\left[
    (\mathbb{I}_{\mathcal{X}} \otimes J_{\Phi_2}) (J_{\Phi_1}^{\mathsf{T}_{\mathcal{Y}}} \otimes \mathbb{I}_{\mathcal{Z}})\right],
\end{align}
where $\star$ denotes the link product and $(\cdot)^{\mathsf{T}_{\mathcal{Y}}}$ denotes the partial transpose with respect to the subsystem $\mathcal{Y}$. 

Analogously, the Choi matrix of a sequential superchannel $\Xi$ is defined by 
\begin{align}
    \mathcal{J}_\Xi = \mathcal{J}_{\Lambda_{n+1}} \star \mathcal{J}_{\Lambda_{n}} \star \cdots \mathcal{J}_{\Lambda_1}, 
\end{align}
where $\Lambda_i$ represent quantum channels appearing in the decomposition~\eqref{eq:Seqdecomp}. The Choi matrix of the output channel $\Phi_{\textrm{out}} = \Xi(\Phi_1, \Phi_2 \cdots \Phi_n)$ is then given by 
\begin{align}
    \mathcal{J}_{\Phi_\textrm{out}} = \mathcal{J}_\Xi \star \qty[\bigotimes_{i=1}^n \mathcal{J}_{\Phi_i}]. 
\end{align}
The Choi matrix of a sequential superchannel $\Xi$ is fully characterized by the conditions
\begin{align}\label{eq:seqChoi}
\begin{dcases}
\mathcal{J}_\Xi \geq 0 \\
\Tr_{I_k, O_k, \cdots I_{n+1}} (\mathcal{J}_\Xi) = \Tr_{O_{k-1}, I_k, O_k, \cdots I_{n+1}} (\mathcal{J}_\Xi)  \otimes \frac{\mathbb{I}_{O_{k-1}}}{d_{O_{k-1}}}\\
\Tr(\mathcal{J}_\Xi) = d_Pd_\mathbf{O}. 
\end{dcases}
\end{align}
for all $k\in[n+1]$. For parallel superchannels, this condition relaxes to: 
\begin{align}\label{eq:parChoi}
\begin{dcases}
\mathcal{J}_\Xi \geq 0 \\
\tr_F(\mathcal{J}_\Xi) = \tr_{\mathbf{O}F}(\mathcal{J}_\Xi) \otimes \frac{\mathbb{I}_\mathbf{O}}{d_{\mathbf{O}}} \\
\tr_{\mathbf{I}\mathbf{O}F}(\mathcal{J}_\Xi) = \tr_{P\mathbf{I}\mathbf{O}F}(\mathcal{J}_\Xi)\otimes \frac{\mathbb{I}_P}{d_{P}} \\
\Tr(\mathcal{J}_\Xi) = d_Pd_{\mathbf{O}}. 
\end{dcases}
\end{align}

\subsection{Schur-Weyl duality}
Consider the following representation of the unitary group $U(d)$ and the symmetric group $\mathfrak{S}_n$,  
\begin{align}
    U^{\otimes n}|i_1\rangle \otimes  \cdots \otimes |i_n\rangle &= U|i_1\rangle \otimes \cdots \otimes U|i_n\rangle\\
    \pi(\sigma) |i_1\rangle \otimes  \cdots \otimes |i_n\rangle &= |i_{\sigma^{-1}(1)} \rangle \otimes \cdots \otimes |i_{\sigma^{-1}(n)} \rangle, 
\end{align}
for $U \in U(d)$ and $\sigma \in \mathfrak{S}_n$. The Schur-Weyl duality is the following decomposition of the Hilbert space: 
\begin{align}
    (\mathbb{C}^{d})^{\otimes n} &\simeq \bigoplus_{\substack{\lambda \vdash n \\\ell(\lambda)\leq d}} \mathcal{U}_\lambda \otimes \mathcal{S}_\lambda \\
    U^{\otimes n} &\simeq \bigoplus_{\substack{\lambda \vdash n \\\ell(\lambda)\leq d}} f_\lambda(U) \otimes \mathbb{I}_{\mathcal{S}_\lambda}\\
    \pi(\sigma) &\simeq \bigoplus_{\substack{\lambda \vdash n \\\ell(\lambda)\leq d}} \mathbb{I}_{\mathcal{U}_\lambda} \otimes g_\lambda(\sigma),
\end{align}
where $\lambda \vdash n$ denotes a partition $(\lambda_1, \lambda_2, \cdots \lambda_{\ell(\lambda)})$ satisfying
$\lambda_1+ \lambda_2+ \cdots \lambda_{\ell(\lambda)}=n$, $\lambda_1 \geq \lambda_2 \geq \cdots \geq \lambda_{\ell(\lambda)}$ and $\ell(\lambda) \leq d$, while $f_\lambda: U(d) \to \mathcal{L}(\mathcal{U}_\lambda)$, $g_\lambda: \mathfrak{S}_n \to \mathcal{L}(\mathcal{S}_\lambda)$ denotes the irreducible representation of $U(d)$ and $\mathfrak{S}_n$ labeled by the partition $\lambda$, respectively. $\mathcal{L}(\mathcal{H})$ denotes the set of linear operators on $\mathcal{H}$. We define the dimensions $\dim \mathcal{U}_\lambda := d_\lambda, \dim \mathcal{S}_\lambda := m_\lambda$ and employ the Gelfand-Tsetlin basis $\{|u\rangle\}_{u\in [d_\lambda]}$~\cite{james2006representation} as an orthonormal basis for $\mathcal{U}_\lambda$ and the Young-Yamanouchi basis $\{|i\rangle\}_{i\in [m_\lambda]}$~\cite{goodman2009symmetry} for an orthonormal basis for $\mathcal{S}_\lambda$. We call the basis $\{|\lambda, u, i \rangle\}_{\lambda \vdash n, \ell(\lambda) \leq d}$ the Schur basis. The unitary transformation taking the computational basis to the Schur basis $U_{\textrm{Sch}}^{(n,d)}$ is called the quantum Schur transform. The Young-Yamanouchi basis is associated to the Standard Young Tableau (SYT). The matrix unit is defined by 
\begin{align}
    E^\lambda_{ij}:= \mathbb{I}_{d_\lambda} \otimes |i\rangle \langle j|
\end{align}
enjoys the following convenient properties
\begin{align}\label{eq:YYformula}
\begin{dcases}
    E^{\alpha}_{ab} \otimes \mathbb{I}_d = \sum_{\mu \in \alpha + \Box} E^{\mu}_{a_\mu b_\mu}\\
    \Tr_{n+1} E_{ij}^\mu  = \delta_{\alpha \beta} \frac{d_\mu}{d_\alpha} E^{\alpha}_{a_\mu b_\mu}
\end{dcases}
\end{align}
where $\alpha+\Box$ is the set of Young diagrams obtained by adding a box to $\alpha$, $\alpha_\mu$ is the SYT obtained by adding the box $n+1$ to the SYT $a$, and $i,j$ are SYT obtained from $a_\mu, b_\mu$ by removing the box $n+1$. In the second equation, the partial trace vanishes if the two tableaux reduce to different shapes after removing $n+1$. 

\section{Rewriting the SDP}\label{ap:SDP}
As discussed in the main text, we consider the performance operator 
\begin{align}
     \widetilde{\Omega}_{n,p} := \frac{1}{d^2}\int dU |U\rangle \rangle \langle \langle U| \otimes \mathcal{J}_{\mathcal{D}_p \circ \mathcal{U}}^{\otimes n}.
\end{align}
For $p=0$, we have 
\begin{align}\label{eq:noerror}
    \widetilde{\Omega}_{k,0} = \frac{1}{d^2} \sum_{\lambda \vdash k+1} \frac{1}{d_{\lambda}}\mathbb{I}_{d_\lambda} \otimes \mathbb{I}_{d_\lambda}  \otimes \sum_{i,j=1}^{m_\lambda} |i\rangle \langle j| \otimes |i\rangle \langle j|. 
\end{align}
For $p>0$, we have 
\begin{align}
    \widetilde{\Omega}_{n,p} = \sum_{\mu, \nu \vdash N} \mathbb{I}_{d_\mu} \otimes \mathbb{I}_{d_\nu} \otimes (\widetilde{\Omega}_{n,p})_{\mu \nu}
\end{align}
due to the unitary symmetry $[\widetilde{\Omega}_{n,p}, V^{\otimes N} \otimes W^{\otimes N}] =0, \, (V,W\in \mathrm{SU}(d))$. Each Schur block of $\widetilde{\Omega}_{n,p}$ can be recursively constructed from $\widetilde{\Omega}_{k,0} \, (1\leq k<n)$ using Eq.~\eqref{eq:noerror} and the branching rule Eq.~\eqref{eq:YYformula}. Now, the SDP of interest for the sequential and parallel case is 
\begin{align}\textbf{\label{eq:SeqChoiap}}
&\max_{\mathcal{J}_\Xi} \Tr(\widetilde{\Omega}^T_{n,p} \mathcal{J}_\Xi) \nonumber\\
&\textrm{subject to}
\begin{dcases}
\mathcal{J}_\Xi \geq 0 \\
\Tr_{I_k, O_k, \cdots I_{n+1}} (\mathcal{J}_\Xi) = \Tr_{O_{k-1}, I_k, O_k, \cdots I_{n+1}} (\mathcal{J}_\Xi)  \otimes \frac{\mathbb{I}_{O_{k-1}}}{d_{O_{k-1}}}\\
\Tr(\mathcal{J}_\Xi) = d_Pd_\mathbf{O}, 
\end{dcases}
\end{align}
and 
\begin{align}\label{eq:ParChoiap}
&\max_{\mathcal{J}_\Xi} \Tr(\widetilde{\Omega}^T_{n,p} \mathcal{J}_\Xi) \nonumber\\
&\textrm{subject to}
\begin{dcases}
\mathcal{J}_\Xi \geq 0 \\
\tr_F(\mathcal{J}_\Xi) = \tr_{\mathbf{O}F}(\mathcal{J}_\Xi) \otimes \frac{\mathbb{I}_\mathbf{O}}{d_{\mathbf{O}}} \\
\tr_{\mathbf{I}\mathbf{O}F}(\mathcal{J}_\Xi) = \tr_{P\mathbf{I}\mathbf{O}F}(\mathcal{J}_\Xi)\otimes \frac{\mathbb{I}_P}{d_{P}} \\
\Tr(\mathcal{J}_\Xi) = d_Pd_{\mathbf{O}}, 
\end{dcases}
\end{align}
respectively. As discussed in the main text, we have 
\begin{align}
\begin{dcases}
    \mathcal{J}_\Xi = \sum_{\mu, \nu \vdash N} \mathbb{I}_{d_\mu} \otimes \mathbb{I}_{d_\nu}  \otimes C_{\mu \nu}\\
     C_{\mu \nu} = \sum_{ijkl}c_{ijkl}^{\mu \nu}|i\rangle \langle j| \otimes |k\rangle \langle l|,
\end{dcases}
\end{align}
due to the unitary symmetry. Here, $C_{\mu \nu} \in \mathrm{End}(\mathcal{S}_\mu \otimes \mathcal{S}_\nu)$ and the systems are ordered $\mathbf{I}P|\mathbf{O}F$. We thus have 
\begin{align}
    \Tr(\widetilde{\Omega}_{n,p} \mathcal{J}_\Xi) = \sum_{\mu, \nu \vdash N} d_\mu d_\nu \Tr(\widetilde{\Omega}_{\mu \nu} C_{\mu \nu})
\end{align}
and $\mathcal{J}_\Xi \geq 0 \Leftrightarrow C_{\mu \nu} \geq 0$ for all $\mu, \nu \vdash N$. Let us discuss the reduction of the remaining trace conditions of ~\eqref{eq:SeqChoiap}, \eqref{eq:ParChoiap} in the Schur basis. \\

\noindent \underline{Sequential comb condition}: The trace conditions in ~\eqref{eq:SeqChoiap} are equivalent to: 
\begin{align}
    \Tr_F(R_{n+1}) &= R_n \otimes \mathbb{I}_{O_n} \label{eq:Seqtr1}\\
    \Tr_{I_t}(R_t) &= R_{t-1} \otimes \mathbb{I}_{O_{t-1}} \quad (t=n, n-1, \cdots 2) \label{eq:Seqtr2}\\
    \Tr_{I_1}(R_1) &= \mathbb{I}_P, \label{eq:Seqtr3}
\end{align}
where $R_{n+1} = \mathcal{J}_\Xi, R_t = \frac{1}{d^{n+1-t}}\Tr_{O_t, I_{t+1} O_{t+1} \cdots I_nO_nF}(\mathcal{J}_\Xi) \: (t=n,n-1 \cdots 1)$. Let us define the branching map $X_{\alpha \nu}:S_\nu \to S_\alpha$, whose matrix elements are given by 
\begin{align}
    [X_{\alpha \nu}]_{ca} = \delta_{c_\nu a}, 
\end{align}
where $\nu \in \alpha+\Box$, $c$ denotes a SYT with frame $\alpha$, $a$ denotes a SYT with frame $\nu$, and $c_\nu$ denotes the SYT with frame $\nu$ obtained by adding a box to the SYT $c$ with frame $\alpha$. Then, the conditions~\eqref{eq:Seqtr1},\eqref{eq:Seqtr2},\eqref{eq:Seqtr3} are equivalent to: 
\begin{align}
\begin{dcases}
    \sum_{\nu \in \alpha+\Box} \frac{d_\nu}{d_\alpha}(\mathbb{I}_\mu\otimes X_{\alpha \nu})C_{\mu \nu} (\mathbb{I}_\mu\otimes X_{\alpha \nu})^\dagger =  \sum_{\beta \in \alpha-\Box} (\mathbb{I}_\mu\otimes X_{\beta \alpha}^\dagger)R_{n, \mu \beta} (\mathbb{I}_\mu\otimes X_{\beta \alpha}^\dagger)^\dagger\\
    \sum_{\mu \in \lambda+\Box} \frac{d_\mu}{d_\lambda}(X_{\lambda \mu} g_\mu(\sigma_t)\otimes \mathbb{I}_{m_\beta})R_{t, \mu \beta} (X_{\lambda \mu} g_\mu(\sigma_t)\otimes \mathbb{I}_{m_\beta})^\dagger  = \sum_{\gamma \in \beta-\Box} (\mathbb{I}_{m_\lambda}\otimes X_{\gamma \beta}^\dagger)R_{t-1, \lambda \gamma} (\mathbb{I}_{m_\lambda}\otimes X_{\gamma \beta}^\dagger)^\dagger\\
    \sum_{\mu \in \lambda+\Box} \frac{d_\mu}{d_\lambda}[X_{\lambda \mu} g_\mu(\sigma_1)\otimes \mathbb{I}_{m_\Box}]R_{1, \mu \Box} [X_{\lambda \mu} g_\mu(\sigma_1)\otimes \mathbb{I}_{m_\Box}]^\dagger = \mathbb{I}_{m_\Box}, 
\end{dcases}
\end{align}
respectively. Here, $\sigma_t$ denotes the representation matrix of the permutation $(t, t+1)$. \\

\noindent \underline{Parallel comb condition}: The trace conditions in~\eqref{eq:ParChoiap} are equivalent to: 
\begin{align}
    \sum_{\nu \in \alpha+\Box} \frac{d_\nu}{d_\alpha}(\mathbb{I}_\mu\otimes X_{\alpha \nu})C_{\mu \nu} (\mathbb{I}_\mu\otimes X_{\alpha \nu})^\dagger = \frac{1}{d^n}\sum_{\nu \vdash N} d_\nu \Tr_{S_\nu}[C_{\mu \nu}]\otimes \mathbb{I}_{m_\alpha}. 
\end{align}
The MATLAB code used to obtain the numerical values is available at~\cite{SDPcode}.

\section{Proof of Eq.~\eqref{eq:compchannelelem}--~\eqref{eq:fidelitylowerbound}}\label{ap:lowerbound}
Here we provide a detailed proof of the equations in the lower bound proof. We begin with the following lemma:
\begin{lm}
\begin{align}\label{eq:sigmatau}
    V^\dagger \sigma_{a,r} V  = \frac{1}{n} \tau_{a}. 
\end{align}
\begin{align}\label{eq:sigmasigmatau}
    V^\dagger \sigma_{a,r} \sigma_{b,s} V = 
    \begin{dcases}
        \delta_{ab} + \frac{i}{n} \epsilon_{abc} \tau_{c} \quad (r=s)\\
        -\frac{1}{n} \delta_{ab} \quad(r\neq s)
    \end{dcases}
\end{align}
\end{lm}
\begin{proof}
A useful starting point is 
\begin{align}
    V_n^\dagger \mathcal{O}  V_n = \frac{1}{m_{\lambda_0}} \Tr_{\mathcal{S}_{\lambda_0}}(\Pi_{\lambda_0}\mathcal{O} \Pi_{\lambda_0})
\end{align}
for any observable $\mathcal{O}\in (\mathbb{C}^2)^{\otimes n}$, where $\Pi_{\lambda_0}$ is the Young projector onto the irrep $\lambda_0$. This implies the permutation symmetry within the codespace, such as 
\begin{align}\label{eq:permsym1}
V^\dagger \sigma_{a,1} V = V^\dagger \sigma_{a,2} V = \cdots V^\dagger \sigma_{a,n} V. 
\end{align}
Next, let us define
\begin{align}
    S_a := \sum_{r=1}^n \sigma_{a,r}. 
\end{align}
$S_a$ in fact acts as a logical Pauli operator. This is because the $\mathrm{SU}(2)$ covariance of the code implies 
\begin{align}
    (e^{i\sigma_{a,1}t} \otimes e^{i\sigma_{a,2}t}  \otimes \cdots \otimes e^{i\sigma_{a,n}t})V = Ve^{i\tau_{a} t}, 
\end{align}
where $\tau_a$ denotes the logical Pauli operator. Differentiating at $t=0$, one obtains 
\begin{align}
    V^\dagger S_a V = \tau_{a} 
\end{align}
Together with the permutation symmetry Eq.~\eqref{eq:permsym1}, we conclude Eq.~\eqref{eq:sigmatau}. Next, let us consider the term $V^\dagger \sigma_{a,r} \sigma_{b,s} V$. For $r=s$, we have 
\begin{align}
    V^\dagger \sigma_{a,r} \sigma_{b,r} V = V^\dagger (\delta_{ab} + i\epsilon_{abc} \sigma_{c,r}) V =  \delta_{ab} + \frac{i}{n} \epsilon_{abc} \tau_{c}. 
\end{align}
For $r\neq s$, we have 
\begin{align}
    V^\dagger \qty(\sum_{r}\sum_{s}\sigma_{a,r} \sigma_{b,s})V = V^\dagger S_a S_b V = \tau_a \tau_b = \delta_{ab}+i\epsilon_{abc}\tau_c. 
\end{align}
Since 
\begin{align}
    V^\dagger \qty(\sum_{r}\sum_{s}\sigma_{a,r} \sigma_{b,s}) V= (n\delta_{ab} + i\epsilon_{abc}\tau_c) +  V^\dagger \qty(\sum_{r\neq s}\sigma_{a,r} \sigma_{b,s}) V,
\end{align}
we have the expression
\begin{align}
    V^\dagger \qty(\sum_{r\neq s}\sigma_{a,r} \sigma_{b,s}) V = -(n-1)\delta_{ab}. 
\end{align}
Due to the permutation symmetry, one has $V^\dagger \sigma_{a,r} \sigma_{b,s} V = V^\dagger \sigma_{a,1} \sigma_{b,2} V$, so 
\begin{align}
    V^\dagger \sigma_{a,r} \sigma_{b,s} V = -\frac{1}{n} \delta_{ab}. 
\end{align}
\end{proof}
\noindent The above lemma gives the matrix elements 
\begin{align}\label{eq:A}
\begin{dcases}
    A_{00} = q \\
    A_{(a,r), 0} = 0 \\
    A_{0, (b,s)} = 0 \\
    A_{(a,r), (b,s)} = 
    \begin{dcases}
        \epsilon_p \delta_{ab} \quad (r=s) \\
        -\frac{\epsilon_p}{n}\delta_{ab} \quad (r\neq s)
    \end{dcases}
\end{dcases}
\end{align}
\begin{align}
\begin{dcases}
    \delta A_{00}^c = 0\\
    \delta A_{(a,r), 0}^c = \frac{\sqrt{q\epsilon_p} }{n} \delta_{ca} \\
    \delta A_{0, (b,s)}^c =  \frac{\sqrt{q\epsilon_p} }{n} \delta_{bc}\\
    \delta A_{(a,r), (b,s)}^c = 
    \begin{dcases}
        -\frac{i\epsilon_p}{n} \epsilon_{abd} \delta_{cd} \quad (r=s) \\
        0 \quad (r\neq s). 
    \end{dcases}
\end{dcases}
\end{align}
Now, let $q:= \qty(1-\frac{3}{4}p)^n, \epsilon_p := \frac{p}{4}\qty(1-\frac{3}{4}p)^{n-1}$. We have 
\begin{align}
    A = q|0\rangle \langle 0| + \mathbb{I}_3 \otimes \qty(\epsilon_p \qty(1+\frac{1}{n}) \mathbb{I}_{n}-\epsilon_p|u\rangle\langle u|), 
\end{align}
where $|u\rangle := \frac{1}{\sqrt{n}} \sum_{r=1}^n |r\rangle$. Thus, we have 
\begin{align}
    A = q|0\rangle \langle 0| + \qty(\sum_{a=x,y,z} |a\rangle \langle a|) \otimes \qty(\frac{\epsilon_p}{n} |u\rangle \langle u| + \epsilon_p \qty(1+\frac{1}{n})\sum_{v \in u^\perp} |v\rangle \langle v|). 
\end{align}
We also have 
\begin{align}
    \frac{1}{2}\sum_c(\delta A)^c \otimes \tau_c^T &=\frac{1}{2}\sqrt{\frac{\epsilon_p q}{n}} \sum_c (|c,u\rangle\langle 0| + |0\rangle \langle c,u|) \otimes \tau_c^T\nonumber\\
    &-\frac{\epsilon_p}{2n}\sum_{a,b,c}i\epsilon_{abc} |a\rangle \langle b| \otimes (|u\rangle \langle u|+\Pi_{u^\perp}) \otimes \tau_c^T. 
\end{align}
Now, we define 
\begin{align}
    W_u := \frac{1}{\sqrt{3}} \sum_{a} |a,u\rangle \otimes \tau_a^T, \quad U_0 := |0\rangle \otimes \mathbb{I}_L
\end{align}
Then, 
\begin{align}
    W_uW_u^\dagger &= \frac{1}{3} \sum_{a,b} |a\rangle \langle b| \otimes |u\rangle \langle u| \otimes \tau_a^T\tau_b^T = \frac{1}{3} \qty(\mathbb{I}_3 \otimes |u\rangle \langle u| \otimes \mathbb{I}_L+\sum_{a,b,c} (-i\epsilon_{abc}|a\rangle \langle b|) \otimes |u\rangle \langle u| \otimes \tau_c^T). 
\end{align}
Therefore, 
\begin{align}
    A \otimes \frac{\mathbb{I}_L}{2} + \frac{1}{2}\sum_c(\delta A)^c \otimes \tau_c^T &= \left[\frac{q}{2} U_0U_0^\dagger + \frac{1}{2}\sqrt{\frac{3q\epsilon_p}{n}} (W_uU_0^\dagger + U_0W_u^\dagger)+ \frac{3\epsilon_p}{2n}W_uW_u^\dagger \right] \nonumber\\
    &+\frac{\epsilon_p}{2}\qty(\frac{n+1}{n}) (\mathbb{I}_3 \otimes \Pi_{u^\perp} \otimes \mathbb{I}_L -\frac{1}{n+1}(\mathbb{I}_3 \otimes \Pi_{u^\perp} \otimes  \mathbb{I}_L-3\sum_v W_vW_v^\dagger)), 
\end{align}
with $\Pi_{u^\perp}:= \sum_{v} |v\rangle \langle v|$. Now, let us choose the fixed state to be 
\begin{align}
    \rho_E =\frac{q}{q+\frac{(n-1)\epsilon_p s_n^2}{3}}|0\rangle \langle 0| + \frac{s_n^2 \epsilon_p}{9\qty(q+\frac{(n-1)\epsilon_ps_n^2}{3})} \mathbb{I}_3 \otimes \Pi_{u^\perp}, 
\end{align}
where $s_n := 2+ \sqrt{\frac{n+3}{n}}$. This is a valid quantum state since $\rho_E \geq 0, \Tr(\rho_E) = 1$. The contribution to the square root fidelity $f_{\mathrm{root}} = \sqrt{f} = \Tr\sqrt{\rho^{\frac{1}{2}}\sigma\rho^{\frac{1}{2}}}$ from the $\mathrm{span} \{ |0\rangle, |u,a\rangle \}$ sector is 
\begin{align}
    2f_{\mathrm{root}}\qty(\mqty(\frac{q}{2} & \frac{1}{2}\sqrt{\frac{3q\epsilon_p}{n}}\\
     \frac{1}{2}\sqrt{\frac{3q\epsilon_p}{n}} & \frac{3\epsilon_p}{2n}), \mqty(\frac{q}{2\qty(q+\frac{(n-1)\epsilon_p s_n^2}{3})} & 0\\
    0 & 0)) = \frac{q}{\sqrt{\qty(q+\frac{(n-1)\epsilon_p s_n^2}{3})}}
\end{align}
and the contribution to the square root fidelity from the $\mathrm{span}\{|v\rangle \}, \, (v \in u^{\perp})$ sector is  
\begin{align}
    (n-1) \times \frac{\frac{ \sqrt{\epsilon_p}s_n}{3}}{\sqrt{2\qty(q+\frac{(n-1)\epsilon_p s_n^2}{3})}}\qty(2\sqrt{\frac{\epsilon_p}{2}\qty(1+\frac{3}{n})}+4\sqrt{\frac{\epsilon_p}{2}}) = \frac{ (n-1) \times  \frac{ \epsilon_p s_n^2}{3}}{\sqrt{\qty(q+\frac{(n-1)\epsilon_p s_n^2}{3})}}. 
\end{align}
These facts together imply the fidelity lower bound 
\begin{align}
    f_{\textrm{av}} &\geq  q+\frac{(n-1)\epsilon_p s_n^2}{3} \nonumber\\
    &= \qty(1-\frac{3}{4}p)^n+\frac{(n-1)}{3} \qty(2+ \sqrt{\frac{n+3}{n}})^2 \frac{p}{4}\qty(1-\frac{3}{4}p)^{n-1}\nonumber\\
    &= 1-\qty(\frac{3}{4}n - \frac{n-1}{3}\qty(1 + \frac{1}{2}\sqrt{\frac{n+3}{n}})^2)p +O(n^2p^2)\nonumber \\
    &=  1- \frac{9p}{8n} + O\qty(\frac{p}{n^2}) + O(n^2p^2)
\end{align}

\section{Concatenation}\label{ap:concatenation}
\subsection{Validity of concatenation schemes}
\begin{lm}
Let $\mathcal{N} = \mathcal{D}_p \circ \mathcal{U}$, and let the output of a twirled purification  superchannel $\Xi$ be $\mathcal{P} = \Xi(\mathcal{N}^{\otimes n})$. Then, $\mathcal{P} = \mathcal{D}_q \circ \mathcal{U}$ for some real number $q$. 
\end{lm}
\begin{proof}
Let $\Xi$ denote the twirled superchannel with the symmetry
\begin{align}
    [V_P^*\otimes V^{\otimes n}_{\mathbf{I}}\otimes W_F^*\otimes W_{\mathbf{O}}^{\otimes n}, \mathcal{J}_\Xi] = 0. 
\end{align}
Let $\mathcal{Q}$ denote the channel 
\begin{align}
    \mathcal{Q} = \Xi(\mathcal{D}_p^{\otimes n}). 
\end{align}
Using the symmetry, we can show that 
\begin{align}
    \Xi([\mathcal{D}_p\circ \mathcal{U}]^{\otimes n}) &= \mathcal{U}\circ \mathcal{U^{\dagger}}\circ \Xi([\mathcal{D}_p\circ \mathcal{U}]^{\otimes n})\nonumber\\
    &= \mathcal{U} \circ \mathcal{Q}. 
\end{align}
We can also show that 
\begin{align}
    \Xi([\mathcal{D}_p\circ \mathcal{U}]^{\otimes n}) &= \Xi([\mathcal{D}_p\circ \mathcal{U}]^{\otimes n}) \circ \mathcal{U}^\dagger \circ \mathcal{U}\nonumber\\
    &= \mathcal{Q} \circ \mathcal{U}. 
\end{align}
Therefore, $\mathcal{Q}$ commutes with an arbitrary unitary $U$, which implies that it is a depolarizing channel. 
\end{proof}

\subsection{Performance of concatenation schemes}
Suppose that a $k$-slot purification superchannel satisfies
\begin{align}
    f_k(p) = 1-C_k^{(1)}p+O_k(p^2).
\end{align}
Then, the effective depolarizing parameter after one layer is
\begin{align}
    p \mapsto r_k p+O_k(p^2),
    \qquad
    r_k := \frac{4C_k^{(1)}}{3}.
\end{align}
After $\ell$ concatenation layers, one has
\begin{align}
    p_\ell = r_k^\ell p+O_k(p^2),
    \qquad
    n=k^\ell.
\end{align}
Therefore,
\begin{align}
    p_n
    \sim
    r_k^{\log_k n}p 
    =
    n^{\log_k r_k}p 
    =
    n^{\log_k\qty(\frac{4C_k^{(1)}}{3})}p.
\end{align}
Hence the concatenated $k$-slot scheme suppresses the noise polynomially in $n$ whenever $\frac{4C_k^{(1)}}{3}<1$. For $k=3$, we have the value $C_3^{(1)} = \frac{5}{4}-\frac{2\sqrt{2}}{3} \simeq 0.3072$. Fitting our numerical data at $p\ll1$ allows us to obtain $C_4^{(1)} = 0.2599, C_5^{(1)} = 0.1968$. We also note that the numerically observed value $C_5^{(1)} \simeq 0.1968$ agrees with our analytically obtained bound $\frac{113}{60}-\frac{8\sqrt{10}}{15} \geq C_5^{(1)}$. Substituting these values gives 
\begin{align}
    p_n \sim n^{-\alpha_k} p, 
\end{align}
where $\alpha_3 = 0.8125, \alpha_4 = 0.7645, \alpha_5 = 0.8313$. Thus, concatenating the $5$-slot purification schemes would be slightly better than concatenating $3$-slot purification schemes. 

\section{Dephasing noise}\label{ap:dephasing}
Here we provide an analysis for the dephasing noise. Analogous to the depolarizing case, we have access to the noisy unitary channel 
\begin{align}
    \widetilde{\mathcal{N}}_{\mathcal{U}, p} = \widetilde{D}_p \circ \mathcal{U}, 
\end{align}
where $\widetilde{\mathcal{D}}_p(\rho) = \qty(1-\frac{p}{2}) \rho + \frac{p}{2}Z\rho Z^\dagger$ is the dephasing channel and $\mathcal{U}$ is an unknown unitary channel. The goal is to find a purification superchannel $\Xi$ that outputs a channel that best approximates $\mathcal{U}$. We use the channel fidelity $f_{\mathrm{Choi}}$ to quantify the performance of the protocol, and use $\widetilde{f}^{\mathrm{Seq}}_{n,d}(p), \widetilde{f}^{\mathrm{Par}}_{n,d}(p)$ to denote the optimal fidelity for purifying $d$-dimensional noisy channels with $n$ uses under sequential and parallel strategies, respectively. 

We first perform a numerical analysis for $d=2, n=1,2,3$. For $n=1,2$, we find that there is no non-trivial purification protocol within sequential strategies. On the other hand, we find that there exists a non-trivial purification protocol for $n\geq 3$. Our numerics suggest that the optimal fidelity is in fact achieved by a parallel protocol for the case $d=2, n=3$: $\widetilde{f}_{3,2}^{\mathrm{Seq}}(p) = \widetilde{f}_{3,2}^{\mathrm{Par}}(p)$. 
\begin{figure}[H]
    \centering
    \includegraphics[scale=0.55]{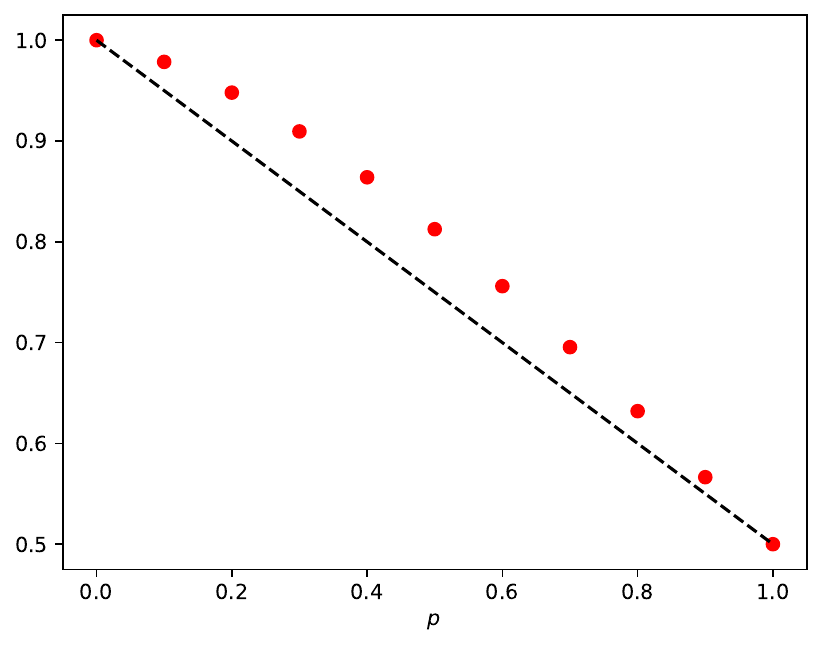}
    \caption{The red dots show the optimal fidelity 
    $\widetilde{f}_{3,2}^{\mathrm{Seq}}(p)$. Our numerics suggest that the optimal fidelity is achieved by a parallel strategy: $\widetilde{f}_{3,2}^{\mathrm{Seq}}(p) = \widetilde{f}_{3,2}^{\mathrm{Par}}(p)$. The black dashed line shows the non-purified trivial fidelity.}
\end{figure}
\noindent We observe that the optimal fidelity obtained by our numerics agrees perfectly with the following analytical form:
\begin{align}
        \qty(1-\frac{p}{2})^3+ \frac{8}{3}\qty(1-\frac{p}{2})^2\qty(\frac{p}{2}) + \frac{1}{3}\qty(1-\frac{p}{2})\qty(\frac{p}{2})^2. 
\end{align}
Indeed, this fidelity is achievable by the following Choi matrix $\mathcal{J}_\Xi$: 
\begin{align}
    &\mathcal{J}_\Xi = (V_F \otimes V_P)\mathcal{J}'_\Xi (V_F \otimes V_P)^\dagger,\\
    &\mathcal{J}'_\Xi = \sum_{i=1}^6 2|K_i\rangle \langle K_i| + \sum_{i=1}^4 |L_i\rangle \langle L_i|, 
\end{align}
where $V = \mqty[ 0 & -1\\ 1 & 0]$, and $|K_i\rangle$ and $|L_i\rangle$ are provided in the Young-Yamanouchi basis as 
\begin{align}
\begin{dcases}
    |K_1\rangle ={}
    \frac{\sqrt2}{2}
    \left|\raisebox{0.2em}{\scalebox{0.2}{\ydiagram{2,2}}},S_1^{(2,2)},T_1^{(2,2)}\right\rangle
    \otimes
    \left|\raisebox{0.2em}{\scalebox{0.2}{\ydiagram{2,2}}},S_1^{(2,2)},T_1^{(2,2)}\right\rangle+
    \frac{\sqrt2}{2}
    \left|\raisebox{0.2em}{\scalebox{0.2}{\ydiagram{2,2}}},S_1^{(2,2)},T_2^{(2,2)}\right\rangle
    \otimes
    \left|\raisebox{0.2em}{\scalebox{0.2}{\ydiagram{2,2}}},S_1^{(2,2)},T_2^{(2,2)}\right\rangle,
    \\
    |K_2\rangle 
    ={}\frac{\sqrt2}{2}
    \left|\raisebox{0.2em}{\scalebox{0.2}{\ydiagram{2,2}}},S_1^{(2,2)},T_1^{(2,2)}\right\rangle
    \otimes
    \left|\raisebox{0.2em}{\scalebox{0.2}{\ydiagram{3,1}}},S_2^{(3,1)},T_1^{(3,1)}\right\rangle+
    \frac12
    \left|\raisebox{0.2em}{\scalebox{0.2}{\ydiagram{2,2}}},S_1^{(2,2)},T_1^{(2,2)}\right\rangle
    \otimes
    \left|\raisebox{0.2em}{\scalebox{0.2}{\ydiagram{3,1}}},S_2^{(3,1)},T_2^{(3,1)}\right\rangle
    \\
    \quad \quad \quad \:-
    \frac12
    \left|\raisebox{0.2em}{\scalebox{0.2}{\ydiagram{2,2}}},S_1^{(2,2)},T_2^{(2,2)}\right\rangle
    \otimes
    \left|\raisebox{0.2em}{\scalebox{0.2}{\ydiagram{3,1}}},S_2^{(3,1)},T_3^{(3,1)}\right\rangle ,
    \\
    |K_3\rangle 
    ={}
    -\frac12
    \left|\raisebox{0.2em}{\scalebox{0.2}{\ydiagram{2,2}}},S_1^{(2,2)},T_1^{(2,2)}\right\rangle
    \otimes
    \left|\raisebox{0.2em}{\scalebox{0.2}{\ydiagram{3,1}}},S_2^{(3,1)},T_3^{(3,1)}\right\rangle+
    \frac{\sqrt2}{2}
    \left|\raisebox{0.2em}{\scalebox{0.2}{\ydiagram{2,2}}},S_1^{(2,2)},T_2^{(2,2)}\right\rangle
    \otimes
    \left|\raisebox{0.2em}{\scalebox{0.2}{\ydiagram{3,1}}},S_2^{(3,1)},T_1^{(3,1)}\right\rangle
    \\
    \quad \quad \quad \:-\frac12
    \left|\raisebox{0.2em}{\scalebox{0.2}{\ydiagram{2,2}}},S_1^{(2,2)},T_2^{(2,2)}\right\rangle
    \otimes
    \left|\raisebox{0.2em}{\scalebox{0.2}{\ydiagram{3,1}}},S_2^{(3,1)},T_2^{(3,1)}\right\rangle,
    \\
    |K_4\rangle 
    ={}
    \frac{\sqrt2}{2}
    \left|\raisebox{0.2em}{\scalebox{0.2}{\ydiagram{2,2}}},S_1^{(2,2)},T_1^{(2,2)}\right\rangle
    \otimes
    \left|\raisebox{0.2em}{\scalebox{0.2}{\ydiagram{4}}},S_3^{(4)},T_1^{(4)}\right\rangle
    -
    \frac12
    \left|\raisebox{0.2em}{\scalebox{0.2}{\ydiagram{2,2}}},S_1^{(2,2)},T_1^{(2,2)}\right\rangle
    \otimes
    \left|\raisebox{0.2em}{\scalebox{0.2}{\ydiagram{2,2}}},S_1^{(2,2)},T_1^{(2,2)}\right\rangle
    \\
    \quad \quad \quad \:+
    \frac12
    \left|\raisebox{0.2em}{\scalebox{0.2}{\ydiagram{2,2}}},S_1^{(2,2)},T_2^{(2,2)}\right\rangle
    \otimes
    \left|\raisebox{0.2em}{\scalebox{0.2}{\ydiagram{2,2}}},S_1^{(2,2)},T_2^{(2,2)}\right\rangle ,
    \\
    |K_5\rangle 
    ={}
    \frac{\sqrt3}{2}
    \left|\raisebox{0.2em}{\scalebox{0.2}{\ydiagram{2,2}}},S_1^{(2,2)},T_1^{(2,2)}\right\rangle
    \otimes
    \left|\raisebox{0.2em}{\scalebox{0.2}{\ydiagram{2,2}}},S_1^{(2,2)},T_2^{(2,2)}\right\rangle+
    \frac{\sqrt6}{6}
    \left|\raisebox{0.2em}{\scalebox{0.2}{\ydiagram{2,2}}},S_1^{(2,2)},T_2^{(2,2)}\right\rangle
    \otimes
    \left|\raisebox{0.2em}{\scalebox{0.2}{\ydiagram{4}}},S_3^{(4)},T_1^{(4)}\right\rangle
    \\
    \quad \quad \quad \:-
    \frac{\sqrt3}{6}
    \left|\raisebox{0.2em}{\scalebox{0.2}{\ydiagram{2,2}}},S_1^{(2,2)},T_2^{(2,2)}\right\rangle
    \otimes
    \left|\raisebox{0.2em}{\scalebox{0.2}{\ydiagram{2,2}}},S_1^{(2,2)},T_1^{(2,2)}\right\rangle ,
    \\
    |K_6\rangle 
    ={}
    \frac{\sqrt3}{3}
    \left|\raisebox{0.2em}{\scalebox{0.2}{\ydiagram{2,2}}},S_1^{(2,2)},T_2^{(2,2)}\right\rangle
    \otimes
    \left|\raisebox{0.2em}{\scalebox{0.2}{\ydiagram{4}}},S_3^{(4)},T_1^{(4)}\right\rangle
    +
    \frac{\sqrt6}{3}
    \left|\raisebox{0.2em}{\scalebox{0.2}{\ydiagram{2,2}}},S_1^{(2,2)},T_2^{(2,2)}\right\rangle
    \otimes
    \left|\raisebox{0.2em}{\scalebox{0.2}{\ydiagram{2,2}}},S_1^{(2,2)},T_1^{(2,2)}\right\rangle, 
\end{dcases}
\end{align}
and 
\begin{align}
\begin{dcases}
    |L_1\rangle 
    ={}
    \left|\raisebox{0.2em}{\scalebox{0.2}{\ydiagram{2,2}}},S_1^{(2,2)},T_2^{(2,2)}\right\rangle
    \otimes
    \left|\raisebox{0.2em}{\scalebox{0.2}{\ydiagram{4}}},S_1^{(4)},T_1^{(4)}\right\rangle ,
    \\[1mm]
    |L_2\rangle 
    ={}
    \left|\raisebox{0.2em}{\scalebox{0.2}{\ydiagram{2,2}}},S_1^{(2,2)},T_1^{(2,2)}\right\rangle
    \otimes
    \left|\raisebox{0.2em}{\scalebox{0.2}{\ydiagram{4}}},S_1^{(4)},T_1^{(4)}\right\rangle ,
    \\[1mm]
    |L_3\rangle 
    ={}
    \left|\raisebox{0.2em}{\scalebox{0.2}{\ydiagram{2,2}}},S_1^{(2,2)},T_2^{(2,2)}\right\rangle
    \otimes
    \left|\raisebox{0.2em}{\scalebox{0.2}{\ydiagram{4}}},S_5^{(4)},T_1^{(4)}\right\rangle ,
    \\[1mm]
    |L_4\rangle 
    ={}
    \left|\raisebox{0.2em}{\scalebox{0.2}{\ydiagram{2,2}}},S_1^{(2,2)},T_1^{(2,2)}\right\rangle
    \otimes
    \left|\raisebox{0.2em}{\scalebox{0.2}{\ydiagram{4}}},S_5^{(4)},T_1^{(4)}\right\rangle. 
\end{dcases}
\end{align}
Here, the definitions of the semistandard Young tableaux and the standard Young tableaux are: 
\begin{align}
\lambda=\scalebox{0.5}{\ydiagram{4}}: \quad 
\ytableausetup{boxsize=0.8em}
\begin{dcases}
S_1^{(4)}=
\begin{ytableau}
1 & 1 & 1 & 1
\end{ytableau},
\\[2mm]
S_2^{(4)}=
\begin{ytableau}
1 & 1 & 1 & 2
\end{ytableau}, 
\\[2mm]
S_3^{(4)}=
\begin{ytableau}
1 & 1 & 2 & 2
\end{ytableau},\quad T_1^{(4)}=
\begin{ytableau}
1 & 2 & 3 & 4
\end{ytableau}
\\[2mm]
S_4^{(4)}=
\begin{ytableau}
1 & 2 & 2 & 2
\end{ytableau},
\\[2mm]
S_5^{(4)}=
\begin{ytableau}
2 & 2 & 2 & 2
\end{ytableau}.
\end{dcases}
\end{align}

\begin{align}
\lambda=\raisebox{0.2em}{\scalebox{0.5}{\ydiagram{3,1}}}: \quad 
\begin{dcases}
S_1^{(3,1)}=
\begin{ytableau}
1 & 1 & 1 \\
2
\end{ytableau}, \quad
T_1^{(3,1)}=
\begin{ytableau}
1 & 2 & 3 \\
4
\end{ytableau}, 
\\[2mm]
S_2^{(3,1)}=
\begin{ytableau}
1 & 1 & 2 \\
2
\end{ytableau}, \quad 
T_2^{(3,1)}=
\begin{ytableau}
1 & 2 & 4 \\
3
\end{ytableau}, \quad 
\\[2mm]
S_3^{(3,1)}=
\begin{ytableau}
1 & 2 & 2 \\
2
\end{ytableau}, \quad 
T_3^{(3,1)}=
\begin{ytableau}
1 & 3 & 4 \\
2
\end{ytableau}.
\end{dcases}
\end{align}

\begin{align}
\lambda=\raisebox{0.2em}{\scalebox{0.5}{\ydiagram{2,2}}}: \quad 
\begin{dcases}
S_1^{(2,2)}=
\begin{ytableau}
1 & 1 \\
2 & 2
\end{ytableau},\quad 
&T_1^{(2,2)}=
\begin{ytableau}
1 & 2 \\
3 & 4
\end{ytableau}\\
&T_2^{(2,2)}=
\begin{ytableau}
1 & 3 \\
2 & 4
\end{ytableau}.
\end{dcases}
\end{align}
One can show that $\mathcal{J}_\Xi$ satisfies the parallel superchannel condition in~\eqref{eq:parChoi}. We note that the encoding channel is nothing but the $U(2)$-covariant entanglement-assisted QECC discussed in our lower bound proof, suggesting that the same code is a natural candidate for an optimal purification of other Pauli noise models as well. The existence of the aforementioned Choi matrix $\mathcal{J}_\Xi$ analytically proves  
\begin{align}
    \qty(1-\frac{p}{2})^3+ \frac{8}{3}\qty(1-\frac{p}{2})^2\qty(\frac{p}{2}) + \frac{1}{3}\qty(1-\frac{p}{2})\qty(\frac{p}{2})^2 \leq \widetilde{f}_{3,2}^{\mathrm{Par}}(p) \leq \widetilde{f}_{3,2}^{\mathrm{Seq}}(p). 
\end{align}
Although our numerics suggest that this lower bound is likely tight for both the parallel and sequential strategies, a rigorous analysis for the tightness is left for future work. The optimal fidelity for $n> 3$, as well as the large-$n$ scaling in the $p\ll 1$ region, also remains open for the dephasing noise. 
\end{document}